\documentclass{article}

\usepackage{graphicx}

\usepackage{lmodern}
\usepackage[left=3.05cm,right=3.05cm,bottom=4.75cm]{geometry}
\usepackage{hyperref}
\usepackage[multiple]{footmisc}

\usepackage[numbers,compress]{natbib}

\usepackage{enumitem}
\usepackage{amssymb,amsfonts,amsthm,mathtools,mathrsfs}
\usepackage{xcolor}
\usepackage{tikz}
\usepackage{caption}
\usepackage{pgfplots}
\pgfplotsset{compat=1.17}
\usetikzlibrary{calc}
\usetikzlibrary{arrows.meta}
\usepackage{ifthen}

\theoremstyle{theorem}
\newtheorem{theorem}{Theorem}[section]
\newtheorem{lemma}[theorem]{Lemma}
\newtheorem{prop}[theorem]{Proposition}
\newtheorem{corollary}[theorem]{Corollary}
\theoremstyle{remark}
\newtheorem{remark}[theorem]{Remark}
\theoremstyle{definition}
\newtheorem{definition}[theorem]{Definition}
\theoremstyle{remark}
\newtheorem{example}[theorem]{Example}

\newcommand{\N}{\mathbb{N}}
\newcommand{\R}{\mathbb{R}}
\newcommand{\Z}{\mathbb{Z}}

\newcommand{\abs}[1]{\left| #1 \right|}
\newcommand{\bigo}{\mathcal{O}}

\newcommand{\im}{\mathrm{im}\,}
\newcommand{\set}[1]{\left\lbrace #1\right\rbrace}
\newcommand{\symdif}{\mathop{\triangle}}

\newcommand{\PS}{\mathcal{P}}
\newcommand{\GF}{\mathbb{Z}_2}

\newcommand{\LS}{\mathcal{L}}
\newcommand{\LSP}{L}
\newcommand{\CS}{\mathcal{C}}

\newcommand{\DHGSE}{\mathscr{D}}
\newcommand{\PGS}{\mathscr{P}}
\newcommand{\PG}{\mathcal{P}}
\newcommand{\PL}{\mathsf{PL}}
\newcommand{\HGSE}{\mathscr{H}}
\newcommand{\HGS}{\mathcal{H}}
\newcommand{\BCS}{\mathcal{B}}

\newcommand{\parmap}{\mathrm{par}}

\newcommand{\lab}{\ell}
\newcommand{\rectangle}{\fboxsep0pt\fbox{\rule{0.75em}{0pt}\rule{0pt}{1ex}}}
\newcommand{\SD}{\mathscr{S}}

\renewcommand\mod{\,\mathrm{mod}\,}

\newcommand{\footremember}[2]{%
   \footnote{#2}
    \newcounter{#1}
    \setcounter{#1}{\value{footnote}}%
}
\newcommand{\footrecall}[1]{%
    \footnotemark[\value{#1}]%
}

\makeatletter
\def\@fnsymbol#1{\ensuremath{\ifcase#1\or 1\or 2\or
   * \or *\or \|\or **\or \dagger\dagger
   \or \ddagger\ddagger \else\@ctrerr\fi}}
    \makeatother

\title{On the uniqueness of compiling graphs\\under the parity transformation}
\author{Florian Dreier\footremember{1}{Institute for Theoretical Physics, University of Innsbruck, Innsbruck A-6020, Austria}$^,$\footremember{2}{Parity Quantum Computing GmbH, Innsbruck A-6020, Austria}$^,$\footremember{3}{florian.dreier@uibk.ac.at}
\and and \and Wolfgang Lechner\footrecall{1} $^,$\footrecall{2}
}

\begin{document}

\maketitle
\begin{abstract}
     In this article, we establish a mathematical framework that utilizes concepts from graph theory to formalize the parity transformation, an encoding strategy for compiling optimization problems on quantum devices. We introduce the transformation as a mapping that encompasses all possible compiled hypergraphs and investigate its uniqueness properties in more detail. Specifically, by introducing so-called loop labelings, we derive an alternative expression of the preimage of any set of compiled hypergraphs under this encoding procedure when all equivalence classes of graphs are being considered. We then deduce equivalent conditions for the injectivity of the parity transformation on any subset of all equivalences classes of graphs. Through concrete examples, we demonstrate that the parity transformation is not an injective mapping, and also introduce an important class of physical layouts and their corresponding set of constraints whose preimage is uniquely determined. In addition, we provide an algorithm which is based on classical algorithms from theoretical computer science and computes a compiled physical layout in this class in polynomial time.\vspace{0.125cm}

    \vspace{0.125cm}

\end{abstract}

\section{Introduction}
The development of quantum computers and quantum algorithms is continuing to advance at full pace and has attracted considerable attention in science and engineering in recent years \cite{OfePetHeeReiLegVla16,AruAryBabBac19,EbaWanLevKeeSemOmr21,JosKokBijKraZacBlaRooZol23,KimEddAnaWeiBerRosNayWuZalTemKan23,MohRamIsaEppPieStrBoiNev23}. As shown in \cite{Sho99} for example, quantum phenomena facilitate the opportunity to develop algorithms for solving certain problems in polynomial time, while no algorithm on classical computers with this time complexity is yet known. However, search problems such as integer factorization are not the only problems where quantum devices can lead to a potential quantum speedup \cite{BleBraCesChoLiPanSum23, KinRayLantHarZuc23, ShaLiChaDeCHer23}. Quantum annealing and the quantum approximate optimization algorithm (QAOA) are well-known techniques \cite{KadNis98,FarGolGutLapLunPre01,JohAmGil11,FarGoldGut14,HauKatLecNishOli20,ZhoWanChoPicLuk20} that are designed to solve combinatorial optimization problems, demonstrating that quantum mechanics can also be applied to complex optimization tasks.
In general, both methods encode a given combinatorial optimization problem into a physical quantum device such that its cost function describes the objective function of the optimization problem itself. One main challenge in the implementation is the exchange of information between the qubits via long-range interactions as well as the scalability of the quantum device as a consequence of the encoding of the optimization problem. Therefore, there is a great need to find encoding strategies which can deal with such issues.
\subsection{The parity architecture}
Extending the so-called LHZ-scheme in \cite{LecHauZol15}, the recently introduced parity architecture \cite{EndHoeNieDriLec23} is an encoding strategy which maps each product of variables in the optimization problem into a new variable, the \emph{parity variable}. The output of an optimization problem under the parity encoding will be a new transformed optimization problem whose objective function consists only of a sum of parity variables. In order to obtain an equivalent optimization problem, constraints to the new optimization problem will be added such that each value of the parity variable equals the corresponding product of values of the original variables.
Therefore, optimization problems are not only mapped to new optimization problems with equal minima but also to instances where the image of all constraint-fulfilling inputs under the new objective function equals the entire image of the original objective function.
For example, the optimization problem
\begin{equation}
    \label{eq:examploptproblem}
    \underset{s_1,\ldots,s_5\in\set{-1,1}}{\min}\sum_{1\leq i<j\leq 5}J_{ij}s_is_j\quad\text{where $J_{ij}\in\R\setminus\set{0}$ for $1\leq i<j\leq 5$}
\end{equation}
can be transformed under the parity encoding into the equivalent optimization problem
\begin{align*}
    \underset{\substack{y_{ij}\in\set{-1,1}\\1\leq i<j\leq 5}}{\min} & \sum_{1\leq i<j\leq 5}J_{ij}y_{ij}\\
    \text{subject to}\quad &\hspace{0.5cm} y_{12}y_{23}y_{13}=1\\
    &\hspace{0.5cm} y_{23}y_{34}y_{24}=1\\
    &\hspace{0.5cm}y_{34}y_{45}y_{35}=1\\
    &\hspace{0.5cm}y_{13}y_{14}y_{24}y_{23}=1\\
    &\hspace{0.5cm}y_{24}y_{25}y_{35}y_{34}=1\\
    &\hspace{0.5cm}y_{14}y_{15}y_{25}y_{24}=1.\\
\end{align*}
Optimization problem in Equation \eqref{eq:examploptproblem} and the constraints of its parity-encoded problem are illustrated in Figure \ref{fig:lhz-scheme}.
\begin{figure}
    \centering
    \begin{center}
\begin{minipage}{0.48\textwidth}
\begin{center}
\includegraphics[]{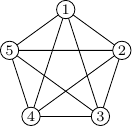}
\end{center}
\end{minipage}
\begin{minipage}{0.48\textwidth}
\begin{center}
\includegraphics[]{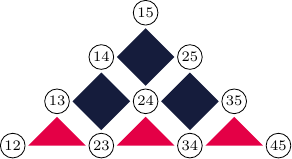}
\end{center}
\end{minipage}
\end{center}
    \caption{Left: Visualization of the original problem as a graph. Right: Illustration of the constraints of the parity-encoded problem which are represented as three- and four-body plaquettes.}
    \label{fig:lhz-scheme}
\end{figure}
We observe that the above constraints correspond to so-called cycles or loops in the graph of the original problem. For example, the first constraint involving the product $y_{12}y_{23}y_{13}$ forms a cycle with the edges $\set{1,2}$, $\set{2,3}$ and $\set{1,3}$ in the original graph. It is worth noting that constraints which are constructed by cycles of the original graph ensure that the transformed problem remains equivalent to the original one. For more details, we refer to \cite{EndHoeNieDriLec23} and Subsection \ref{subsec:paritygraphmapping}. We also emphasize that the construction of the constraints for any combinatorial problem is independent of the coupling strengths in front of the products such as the numbers $J_{ij}$ in \eqref{eq:examploptproblem}. Since in general different sets of constraints for the transformed optimization problem in the parity encoding can be chosen, the original problem can be in principle mapped to different parity-encoded optimization problems. We note that in this paper the parity transformation is treated as a mapping whose output equals the set of equivalence classes of all compiled hypergraphs and therefore takes into account every possible choice for the set of constraints (see Subsection \ref{subsec:paritygraphmapping} for more details).

With this encoding, the  optimization problem can be mapped onto a quantum device where interactions between physical qubits are local. Specifically, by adding ancillas, i. e. parity variables which do not represent a product in the original problem, any optimization problem can be mapped to a new optimization problem which allows to implement all required constraints on a corresponding physical device with local three- and four-body interactions \cite{HoeNieKalLec23}. Moreover, the parity architecture provides an encoding paradigm which can be used independently of the quantum platform such as trapped ions \cite{CirZol95,HafRooBla08}, superconducting qubits \cite{WalSchBlaFruHuaMajKumGirSch04,KjaSchBraKraWanGusOli20} or neutral atoms \cite{BarLieLesLahBro18,LevKee19,HenBegSigLahBroReyJur20,EbaKeeCaiWan22}.

Besides recently developed compilation strategies for the parity transformation \cite{HoeMesLec23,HoeNieKalLec23}, the encoding procedure has also been further extensively studied in the field of quantum optimization as well as quantum computation.
For example in \cite{EndMesFelDlaLec22}, a modular parallelization method for QAOA and parity-encoded optimization problems has been introduced. Furthermore, in \cite{DlaEndMbeKruLecBij22,LanDlaEndLec23} novel techniques for realizing quantum devices to perform quantum optimization with neutral atom platforms using the parity architecture have been proposed. As demonstrated in \cite{FelMesEndLec22}, the parity transformation also enables a new universal quantum computing approach which leads to advantages for crucial quantum algorithms such as the quantum Fourier transform. However, a pivotal practical question, directly related to transformation itself and yet to be thoroughly examined, is under which circumstances non-equivalent optimization problems (those that cannot be transformed into one another by re-labeling) can or cannot be solved with the same compiled physical layout.
Or in other words, how can the set of optimization problems that can be compiled to a given hypergraph under the parity encoding be described and when is the set uniquely determined? Moreover, does there exist a class of physical layouts which do have this uniqueness property? In this paper, we study such questions in more detail and investigate the case where the products in the optimization problems consist of two-body interactions.
\subsection{Outline}
    In Section \ref{sec:notationauxresults}, we introduce some notations and derive auxiliary results which we will use throughout the whole article. Furthermore, by representing the spin variables as vertices and the monomials as an edge of a hypergraph, we establish a mathematical framework to define the parity transformation as a mapping whose domain corresponds to the set of all equivalence classes of hypergraphs. Each equivalence class represents an optimization problem (apart from the individual prefactors such as $J_{ij}$ in \eqref{eq:examploptproblem}) and all its equivalent representations. The output of the parity transformation describes the set of equivalences classes of all possible compiled hypergraphs which can %be implemented on a quantum device.
    be used for designing physical layouts for quantum devices where the vertices represent the physical qubits and the edges their physical interactions. In Section \ref{sec:uniquenesssimplegraph}, we investigate the parity map in more detail and derive a description of the preimage of any set of equivalence classes of compiled hypergraphs, yielding an equivalent condition for the injectivity of the parity transformation on any subdomain. The main statement of Section \ref{sec:uniquenesssimplegraph} is given in Theorem \ref{thm:mainthm}. The last section of this paper applies the derived results of Section \ref{sec:uniquenesssimplegraph} and shows concrete examples which demonstrate that the parity encoding is in general not a unique mapping. However, as we will prove in Subsection \ref{subsec:charrectplaqlayouts}, there exists an important class of compiled hypergraphs called \emph{rectangular plaquette layouts}, whose corresponding optimization problems can be uniquely determined. This class represents a widely encountered physical hardware layout found in real quantum devices, where qubit connectivity is local and follows a square-grid-like structure. As a consequence of Theorem \ref{thm:preimagerectangularlayouts}, we present a procedure that is based on well-known algorithms and can decide in polynomial time whether a given optimization problem can be compiled on a plaquette layout of this type and simultaneously generates the corresponding physical layout.
    
    Except for Section \ref{sec:notationauxresults}, all results will be derived for optimization problems which can be represented as a graph where all edges contain two vertices. While hypergraphs can represent more general cases, this work focuses on graphs due to their fundamental structure and wide applicability.
\section{Notation and auxiliary results}
\label{sec:notationauxresults}
In the first two definitions, we recall two essential terms from graph theory, which are hypergraphs and isomorphic hypergraphs. Note that throughout the whole article, we denote by $f(A)=\set{f(a)\mid a\in A}$ the image of a subset $A\subset X$ and $f^{-1}(B)=\set{x\in X\mid f(x)\in B}$ the preimage of a subset $B\subset Y$ under a function $f\colon X\to Y$.
\begin{definition}[Hypergraph]
    A hypergraph $H$ is a pair $(V,E)$, where $V=\set{v_1,\ldots,v_n}$ is the set of all \emph{vertices} and $E\subset\PS(V)\setminus\set{\emptyset}$ is a subset of the power set of $V$, called the \emph{edge set} of the hypergraph $H$. We call $H$ a \emph{graph}, if for all edges $e\in E$ the number of vertices in $e$ is equal to two.
\end{definition}
\begin{definition}[Isomorphic hypergraphs]
    Let $H=(V,E)$ and $H'=(V',E')$ be two hypergraphs. Then, we say \emph{$H$ is isomorphic to $H'$} if and only if there exists a bijection $f\colon V\to V'$ such that for all $e\subset V$
    \begin{equation}
        \label{eq:isomhypergraphs}
        e\in E\Longleftrightarrow f(e)=\set{f(v)\mid v\in e}\in E'.
    \end{equation}
\end{definition}
\begin{remark}
    \label{rem:isomhyp}
    \begin{enumerate}
        \item Let $\HGS\coloneqq\set{H\mid H\text{ is a hypergraph}}$ be the set of all hypergraphs. Then, the relation
        \begin{equation*}
            H\sim_\HGS H':\Longleftrightarrow H\text{ is isomorphic to } H'
        \end{equation*}
        is an equivalence relation on $\HGS$.
        \item \label{item:isomhyp} For a hypergraph $H=(V,E)$ and a bijection $f\colon V\to V'$ between $V$ and another set $V'$, the hypergraph $H_f=(V_f,E_f)$ with $V_f\coloneqq V'$ and $E_f\coloneqq\set{f(e)\mid e\in E}$ defines an isomorphic hypergraph to $H$, called \emph{the isomorphic hypergraph of $H$ induced by $f$}.
    \end{enumerate}
\end{remark}
\noindent In the next definition, we recall the term \emph{walk} and \emph{path} in graphs and also define \emph{cycles} in graphs (see also \cite{GroYelAnd18}, for example)
\begin{definition}[Walk, path and cycles in graphs]
\label{def:cycles}
    Let $H=(V,E)$ be a graph, $v_0,\ldots,v_n\in V$ and $e_1,\ldots,e_n\in E$.
    \begin{enumerate}
        \item We call the finite tupel $(v_0,e_1,v_1,e_2,\ldots, e_n,v_n)$ a \emph{walk} in the graph $H$ if it satisfies the following two conditions:
    \begin{enumerate}
        \item For all $i=1,\ldots,n$ it holds $v_{i-1}\neq v_i$.
        \item For all $i=1,\ldots,n$ we have $v_{i-1},v_i\in e_i$.
    \end{enumerate}
    \item If $(v_0,e_1,v_1,e_2,\ldots, e_n,v_n)$ is a walk in $H$ and neither a vertex nor an edge is repeated, we call the walk a \emph{path} in $H$. Moreover, we say that a set of edges $E'\subset E$ is a path in $H$ if and only if there exists an enumeration $E'=\set{e_1,\ldots,e_n}$ and $v_0,\ldots,v_n\in V$ such that $(v_0,e_1,v_1,e_2,\ldots, e_n,v_n)$ is a path in $H$.
    \item We call a walk $(v_0,e_1,v_1,e_2,\ldots, e_n,v_n)$ a \emph{cycle} in $H$ if $v_0,\ldots,v_{n-1}$ are pairwise different and $v_0=v_n$. In analogous way for paths, we call a set of edges $C\subset E$ a cycle if  there exists an enumeration $C=\set{e_1,\ldots,e_n}$ and $v_0,\ldots,v_n\in V$ such that $(v_0,e_1,v_1,e_2,\ldots, e_n,v_n)$ is a cycle in $H$.
    \end{enumerate}
\end{definition}
\subsection{Edge space and cycle space}
In this section, we will introduce the edge space of hypergraphs and the cycle space of graphs which will be important for the definition of the constraint space in the next subsection. We will also derive some properties of cycle spaces and will make use of them in Section \ref{sec:uniquenesssimplegraph}.
\begin{definition}
    Let $H=(V,E)$ be a hypergraph. The \emph{edge space of $H$} is defined as the power set $\PS(E)$ of $E$. By using the symmetric difference operator
    \begin{equation*}
        \triangle\colon\PS(E)\times\PS(E)\to\PS(E)\colon (E_1,E_2)\mapsto E_1\triangle E_2
    \end{equation*}
    as a group operation on $\PS(E)$ and defining the scalar multiplication
    \begin{equation*}
        \cdot\colon\GF\times\PS(E)\to\PS(E)\colon (c,E')\mapsto cE'\coloneqq\begin{cases}
            E',\quad & \text{if $c=1$,}\\
            \emptyset,\quad & \text{if $c=0$},
        \end{cases}
    \end{equation*}
    the triple $(\PS(E),\triangle,\cdot)$ is a vector space over the finite field $\GF$.
 \end{definition}
\begin{definition}
\label{def:ls}
For a graph $H=(V,E)$ we define by $\LS_H\coloneqq\mathrm{span}_{\GF}(\LSP_H)$ the \emph{cycle space} or \emph{loop space} of $H$, where \[\LSP_H\coloneqq\set{\set{e_1,\ldots,e_n}\in \PS(E)\mid \exists v_0,\ldots,v_n\in V\colon (v_0,e_1,v_1,e_2,\ldots, e_n,v_n)\text{ is a cycle in $H$}}\] denotes the spanning set of $\LS_H$. 
\end{definition}
\begin{example}
    \label{examp:cycles}
    Let $H=(V,E)$ be the graph with vertex set $V=\set{1,2,3,4,5}$ and edge set $E=\set{\set{1,2},\set{1,4},\set{1,5},\set{2,3},\set{3,4},\set{3,5}}$. Then the cycle space is given by \[\LS_H=\set{\set{\set{1,2},\set{2,3},\set{3,4},\set{4,1}},\set{\set{1,4},\set{4,3},\set{3,5},\set{5,1}},\set{\set{1,2},\set{2,3},\set{3,5},\set{5,1}},\emptyset}.\] All three non-empty cycles of the loop space are illustrated in Figure \ref{fig:cycles}.
    \begin{figure}
        \centering
        \begin{minipage}{0.32\textwidth}
\begin{center}
\includegraphics{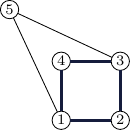}
\end{center}
\end{minipage}
\begin{minipage}{0.32\textwidth}
\begin{center}
\includegraphics{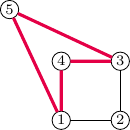}
\end{center}
\end{minipage}
\begin{minipage}{0.32\textwidth}
\begin{center}
\includegraphics{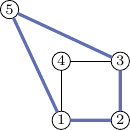}
\end{center}
\end{minipage}
        \caption{Cycles in the loop space of $H$ in Example \ref{examp:cycles}.}
        \label{fig:cycles}
    \end{figure}
\end{example}
\begin{remark}
    \label{rem:ls}
    \begin{enumerate}
        \item By definition, the cycle space $\LS_H$ is a linear subspace of $\PS(E)$ over the finite field $\GF$.
        \item \label{item:subsetls} For a hypergraph $H'=(V',E')$ with $\bigcup_{C\in \LS_H}C\subset E'$, we have $\LS_H\subset\LS_{H'}$. Moreover, in general we have $\LS_H\subset\LS_{H'}$ if the edge set of $H'$ satisfies $E\subset E'$. 
        \begin{proof}
            From $\bigcup_{C\in \LS_H}C\subset E'$ as well as the definition of cycles and the spanning set $\LSP_H$, we see that $\LSP_H\subset \LSP_{H'}$ and therefore, $\LS_H=\mathrm{span}_{\GF}(\LSP_H)\subset \mathrm{span}_{\Z_2}(\LSP_{H'})=\LS_{H'}$. The second statement follows from the the relation $\LSP_H\subset \LSP_{H'}$.
        \end{proof}
        \item \label{item:isomls} If $H'=(V',E')$ is another hypergraph which is isomorphic to $H$, then the cycle spaces $\LS_H$ and $\LS_{H'}$ are isomorphic. For a bijection $f\colon V\to V'$ satisfying \eqref{eq:isomhypergraphs}, an isomorphism between the cycle spaces $\LS_H$ and $\LS_{H'}$ is given by
        \begin{equation*}
            f_{\LS}\colon\LS_H\to\LS_{H'}\colon C\mapsto\bigcup_{e\in C}\set{f(e)}.
        \end{equation*}
        Furthermore, it holds $(f_{\LS})^{-1}=(f^{-1})_{\LS}$.
        \begin{proof}
            \begin{enumerate}[label=(\alph*),wide=\parindent,leftmargin=0pt,align=left]
                \item \label{item:firststepisomcs} As a first step, we prove that
                \begin{equation}
                    \label{eq:cpisom}
                    \LSP_{H'}=\set{\set{f(e)\mid e\in C}\mid C\in \LSP_{H}},
                \end{equation}
            where $\LSP_{H}$ and $\LSP_{H'}$ are the spanning sets of the cycle spaces $\LS_H$ and $\LS_{H'}$, respectively.
            First, assume $C'\in \LSP_{H'}$. Then, there exist vertices $v'_0,\ldots,v'_n\in V'$ and edges $e'_1,\ldots,e'_n\in E'$ such that $(v'_0,e'_1,v'_1,e'_2,\ldots, e'_n,v'_n)$ is a cycle in $H'$ and $C'=\set{e'_1,\ldots,e'_n}$. Define $v_0\coloneqq f^{-1}(v'_0),\ldots,v_n\coloneqq f^{-1}(v'_n)$ and $e_1\coloneqq f^{-1}(e'_1),\ldots,e_n\coloneqq f^{-1}(e'_n)$. Then, from the bijectivity of $f$, we obtain that $(v_0,e_1,v_1,e_2,\ldots,e_n,v_n)$ is cycle in $H$, and therefore \[C\coloneqq\set{e_1,\ldots,e_n}=\set{f^{-1}(e_1),\ldots,f^{-1}(e_n)}\in \LSP_H,\] which shows $C'=\set{f(e)\mid e\in C}$, proving that $\LSP_{H'}$ is a subset of the right set in \eqref{eq:cpisom}. Using similar arguments, it can be shown that $\set{f(e)\mid e\in C}\in \LSP_{H'}$ for $C\in \LSP_{H}$.\\
                \item \label{item:secondstepisomcs}Next, we show that $f_{\LS}(dC_1\triangle C_2)=df_{\LS}(C_1)\triangle f_{\LS}(C_2)$ for $d\in\GF$ and $C_1,C_2\in \LSP_H$. We observe
                \begin{align*}
                    f_{\LS}(dC_1\triangle C_2)&=\bigcup_{e\in dC_1\setminus C_2}\set{f(e)}\cup\bigcup_{e\in C_2\setminus dC_1}\set{f(e)}\\
                    &=d\bigcup_{e\in C_1}\set{f(e)}\setminus\bigcup_{e\in C_2}\set{f(e)} \cup\bigcup_{e\in C_2}\set{f(e)}\setminus d\bigcup_{e\in C_1}\set{f(e)}=df_{\LS}(C_1)\triangle f_{\LS}(C_2).
                \end{align*}
                This shows together with \ref{item:firststepisomcs} that $f_{\LS}$ is a well-defined linear mapping.
                \item From the bijectivity of $f$, we immediately observe that $f_{\LS}$ is injective. To prove surjectivity, assume $C'\in\LS_{H'}$ and define $C\coloneqq \bigcup_{e'\in C'}\set{f^{-1}(e')}$. By replacing $f$ with $f^{-1}$ and using the same calculations as in \ref{item:firststepisomcs} and \ref{item:secondstepisomcs}, we can deduce that $C\in\LS_{H}$. The relations $f_{\LS}(C)=C'$ and $(f_{\LS})^{-1}=(f^{-1})_{\LS}$ can be easily verified.\qedhere
            \end{enumerate}
        \end{proof}
        \item \label{item:csbases} By $\BCS_{\LS_H}\coloneqq\set{B\subset\LS_H\setminus\set{\emptyset}\mid\text{$B$ is basis of $\LS_H$}}$ we denote the set of all bases of $\LS_H$.
        \item As mentioned in \cite{LieRiz07}, there exist different classes of cycle bases of graphs, which we will also encounter in this article. To name two of these, the \emph{fundamental basis} class of a graph $H$ consists of all cycle bases $B\in\BCS_{\LS_H}$ with $B\subset \LSP_H$ which satisfy
        \begin{equation*}
            \forall C\in B\colon C\setminus\bigcup_{C'\in B\setminus\set{C}} C'\neq\emptyset,
        \end{equation*}
        whereas for bases in the \emph{weakly fundamental basis} class there exists an enumeration of $B=\set{B_1,\ldots,B_n}$ which only satisfy
        \begin{equation*}
            \forall 2\leq k\leq n: B_k\setminus\bigcup_{l=1}^{k-1} B_l\neq\emptyset.
        \end{equation*}
    \end{enumerate}
\end{remark}
\noindent Before we continue with the definition of the constraint space of hypergraphs, we present a well-known result for the cycle space of graphs, which will be used in %the last section of this article.
Section \ref{sec:applications}. It shows a formula for the dimension of the cycle space for graphs (see \cite{GroYelAnd18}, for example).
\begin{theorem} \label{thm:dimformulals}Let $H=(V,E)$ be a graph and $c$ denote the number of components of the graph $H$. Then, it holds \[\dim(\LS_{H})=\abs{E}-\abs{V}+c.\]
\end{theorem}
\subsection{The constraint space of hypergraphs}
Subsequent to the previous section, we are now able to define the so-called constraint space of hypergraphs. As we will later see, this space will be one of the main objects appearing in the definition of the parity transformation in Subsection \ref{subsec:paritygraphmapping}.
\begin{definition}[Constraint space of hypergraphs]
    Let $H=(V,E)$ be a hypergraph and $C\in\PS(E)$ be an element of the edge space of $H$. We call $C$ a \emph{constraint of the hypergraph $H$} if it satisfies the following condition:
    \begin{equation*}
        \forall v\in V_C\colon \abs{\set{e\in C\mid v\in e}}\text{ is even}, 
    \end{equation*}
    where $V_C\coloneqq\set{v\in V\mid \exists e\in C\colon v\in e}$. Furthermore, we denote by $\CS_H\coloneqq\set{C\in\PS(E)\mid C\text{ is a constraint}}$ the \emph{constraint space of $H$}.
\end{definition}
\begin{example}
    \label{examp:constraints}
    Let $H=(V,E)$ with $V=\set{1,2,3,4,5}$ and $E=\set{\set{1,2},\set{2,5},\set{1,3},\set{1,2,4},\set{3,4,5}}$. Define $C_1\coloneqq\set{\set{2,5},\set{1,3},\set{1,2,4},\set{3,4,5}}$ and $C_2\coloneqq\set{\set{1,2},\set{2,5},\set{1,3},\set{3,4,5}}$. Then we see that $V_{C_1}=V$ as well as $\abs{\set{e\in C_1\mid v\in e}}=2$ for all $v\in V_{C_1}$, which implies that $C_1$ is a constraint. However, $\abs{\set{e\in C_2\mid 4\in e}}=1$ which shows that $C_2$ is not a constraint in $H$. Figure \ref{fig:constraints} illustrates the two elements $C_1$ and $C_2$ in the edge space of $H$.
    \begin{figure}
        \centering
        \begin{minipage}{0.48\textwidth}
\begin{center}
\includegraphics{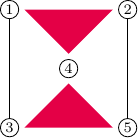}
\end{center}
\end{minipage}
\begin{minipage}{0.48\textwidth}
\begin{center}
\includegraphics{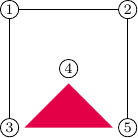}
\end{center}
\end{minipage}
        \caption{Left: Constraint $C_1$ in Example \ref{examp:constraints}. Right: Illustration of $C_2\in\PS(E)$ in Example \ref{examp:constraints} which is not a constraint.}
        \label{fig:constraints}
    \end{figure}
\end{example}
\begin{remark}
    \label{rem:cs}
    Let $H=(V,E)\in\HGS$.
    \begin{enumerate}
        \item The constraint space $\CS_H$ is a linear subspace of $\PS(E)$.
        \begin{proof}
            Let $C_1,C_2\in\CS_H$ and $d\in\GF$. Then, we observe from the definition of the symmetric difference operator for $v\in V_{dC_1\triangle C_2}$
            \begin{align*}
                |\{e\in dC_1\triangle C_2\mid\,&v\in e\}|\\
                &=\abs{\set{e\in dC_1\setminus C_2\mid v\in e}}+\abs{\set{e\in C_2\setminus dC_1\mid v\in e}}\\
                &=\abs{\set{e\in dC_1\mid v\in e}\setminus\set{e\in C_2\mid v\in e}}+\abs{\set{e\in C_2\mid v\in e}\setminus\set{e\in dC_1\mid v\in e}}\\
                &=\abs{\set{e\in dC_1\mid v\in e}}+\abs{\set{e\in C_2\mid v\in e}}-2\abs{\set{e\in dC_1\cap C_2\mid v\in e}}.
            \end{align*}
            By assumption, we have $\abs{\set{e\in dC_1\mid v\in e}}$ and $\abs{\set{e\in C_2\mid v\in e}}$ is even. Thus, $|\{e\in dC_1\triangle C_2\mid v\in e\}|$ is even and hence, $dC_1\triangle C_2\in\CS_H$.
        \end{proof}
        \item \label{item:csequalls} If $H$ is a graph, then $\CS_H=\LS_H$. %Therefore, the constraint space can be seen as a variant of an extension of the cycle space to hypergraphs.
        \begin{proof}
            This follows from the application of Euler's theorem (see Theorem 4.5.11 in \cite{GroYelAnd18}, for example) on each component of the graph $H'=(V',E')$, where $E'$ is a constraint in $\CS_H$ and $V'=\set{v'\in V'\mid\exists e'\in E'\colon v'\in e'}$.
        \end{proof}
        \item \label{item:isomcs} If $H'=(V',E')$ is another hypergraph which is isomorphic to $H$, then the constraint spaces $\CS_H$ and $\CS_{H'}$ are isomorphic. For a bijection $f\colon V\to V'$ satisfying \eqref{eq:isomhypergraphs}, an isomorphism between the constraint spaces $\CS_H$ and $\CS_{H'}$ is given by
        \begin{equation*}
            f_{\CS}\colon\CS_H\to\CS_{H'}\colon C\mapsto\bigcup_{e\in C}\set{f(e)}.
        \end{equation*}
        \begin{proof}
            First, we show that $f_{\CS}$ is well-defined. Let $C\in\CS_H$ and $C'\coloneqq f_{\CS}(C)$. For $v'\in V_{C'}$ we observe that
            \begin{multline*}
                \abs{\set{e'\in C'\mid v'\in e'}}\\
                =\bigg|\biguplus_{\substack{e'\in C'\\ v'\in e'}} \set{e'}\bigg|=\bigg|\biguplus_{\substack{e\in C\\ v'\in f(e)}} \set{f(e)}\bigg|=\sum_{\substack{e\in C\\ v'\in f(e)}}\abs{\set{f(e)}}=\sum_{\substack{e\in C\\ f^{-1}(v')\in e}}\abs{\set{e}}\\=\abs{\set{e\in C\mid f^{-1}(v')\in e}}
            \end{multline*}
            and from the definition of $V_{C'}$ we see that $f^{-1}(v')\in V_C$. Therefore, $\abs{\set{e'\in C'\mid v'\in e'}}$ is even, which shows that $f_{\CS}(\CS_H)\subset\CS_{H'}$. Next, we define $C\coloneqq \bigcup_{e'\in C'}\set{f^{-1}(e')}$ for some $C'\in\CS_{H'}$. Using the same calculations as before, we can deduce that $C\in\CS_H$. Similarly as in the proof of Remark \ref{rem:ls} \ref{item:isomls}, we then obtain $f_{\CS}(C)=C'$, which shows that $f_{\CS}$ is surjective. The linearity and injectivity of $f_{\CS}$ follows from the same arguments as in Remark \ref{rem:ls} \ref{item:isomls}.\qedhere
        \end{proof}
        \item Analogously to the cycle space, we define $\BCS_{\CS_H}$ as the set of all bases of the constraint space of $H$.
    \end{enumerate}
\end{remark}
\noindent As we can observe in Remark \ref{rem:cs} \ref{item:csequalls}, the constraint space can be seen as a variant of an extension of the cycle space for hypergraphs. In \cite{Ber84}, the definition of cycles for arbitrary hypergraphs has been extended, and it is almost identical to the one as for graphs in Definition \ref{def:cycles}. For hypergraphs, it is additionally assumed that edges in the cycle are pairwise distinct, which is obviously true for edges in a graph. However, using this definition of cycles for hypergraphs and defining the corresponding cycle space as in Definition \ref{def:ls} can lead to a space which is not equal to the constraint space as the following example shows.
\begin{example}
    Let $H=(V,E)$ be the hypergraph as defined in Example \ref{examp:constraints}. Then, we observed that $C_2$ is not an element in the constraint space but using the sequence \[(1,\set{1,2},2,\set{2,5},5,\set{3,4,5},3,\set{1,3},1)\] we deduce that $C_2$ is in the cycle space of $H$. This shows that every cycle need not to be a constraint of a hypergraph. Conversely, let $H'=(V',E')$ be the hypergraph with $V'\coloneqq\set{1,2,3,4}$ and $E'\coloneqq\set{\set{1,2,3,4},\set{1,2},\set{3,4}}$. Then, we observe that \[\LS_H=\set{\set{\set{1,2},\set{1,2,3,4}},\set{\set{3,4},\set{1,2,3,4}},\set{\set{1,2},\set{3,4}},\emptyset}.\]
    However, the edge set $E'$ is a constraint of $H$ but not an element of the cycle space of $H$. Therefore, there also exist constraints of hypergraphs which are not contained in the cycle space.
\end{example}
\subsection{The parity graph mapping}
\label{subsec:paritygraphmapping}
For the definition of the parity transformation, which we will introduce as the parity graph mapping, we will make further use of the following notations.
\begin{definition}
    \label{def:comphyps}
    Let $H=(V,E)\in\HGS$ and $E_B\coloneqq \set{e\in E\mid \exists C\in B\colon e\in C}$ be the set of edges of a basis $B\in\BCS_{\CS_H}$ with cardinality $m\coloneqq\abs{E_{B}}$. Then, for an enumeration $h\colon\set{1,\ldots,m}\to E_B$ of the set $E_B$, we denote by $\PG_{B}=(V_{\PG_{B}},E_{\PG_{B}})\in\HGS$ with the vertex set $V_{\PG_{B}}\coloneqq\set{1,\ldots,m}$ and the edge set \[E_{\PG_{B}}\coloneqq\bigcup_{C\in B}\set{h^{-1}(C)}\] the \emph{compiled hypergraph of the basis $B$ in $\BCS_{\CS_H}$}. Moreover, we call \[\PG_H\coloneqq \set{[\PG_{B}]\mid B\in\BCS_{\CS_H}}\] the \emph{set of compiled hypergraphs of $H$}.
\end{definition}
\begin{example}
    \label{examp:comphypergraphs}
    Let $H=(V,E)\in\HGS$ be the graph defined by $V=\set{1,2,3,4}$ and\\$E=\set{\set{1,2},\set{2,3},\set{3,4},\set{1,4},\set{2,4}}$. We observe that
    \begin{align*}
        B_1&=\set{\set{\set{1,2},\set{2,3},\set{3,4},\set{4,1}},\set{\set{1,2},\set{2,4},\set{4,1}}},\\
        B_2&=\set{\set{\set{1,2},\set{2,3},\set{3,4},\set{4,1}},\set{\set{2,3},\set{3,4},\set{4,2}}},\\
        B_3&=\set{\set{\set{1,2},\set{2,4},\set{4,1}},\set{\set{2,3},\set{3,4},\set{4,2}}}\quad\text{and}\\
        \BCS_{\CS_H}&=\set{B_1,B_2,B_3}.
    \end{align*}
    Next, we choose the enumeration $h\colon\set{1,2,3,4,5}\to E$ according to the table
    \begin{center}
        \begin{tabular}{c | c | c | c | c | c}
            $i$ & $1$ & $2$ & $3$ & $4$ & $5$\\\hline
            $h(i)$ & $\set{1,2}$ & $\set{2,3}$ & $\set{3,4}$ & $\set{4,1}$ & $\set{2,4}$
        \end{tabular}
    \end{center}
    yielding the compiled hypergraphs
    \begin{align*}
        &\PG_{B_1}=(V_{\PG_{B_1}},E_{\PG_{B_1}})\quad\text{ with }\quad V_{\PG_{B_1}}=\set{1,2,3,4,5}\quad\text{and}\quad E_{\PG_{B_1}}=\set{\set{1,2,3,4},\set{1,4,5}},\\
        &\PG_{B_2}=(V_{\PG_{B_2}},E_{\PG_{B_2}})\quad\text{ with }\quad V_{\PG_{B_2}}=\set{1,2,3,4,5}\quad\text{and}\quad E_{\PG_{B_2}}=\set{\set{1,2,3,4},\set{2,3,5}},\\
        &\PG_{B_3}=(V_{\PG_{B_3}},E_{\PG_{B_3}})\quad\text{ with }\quad  V_{\PG_{B_3}}=\set{1,2,3,4,5}\quad\text{and}\quad E_{\PG_{B_3}}=\set{\set{1,4,5},\set{2,3,5}}.
    \end{align*}
    Figure \ref{fig:comphypergraphs} illustrates the three compiled hypergraphs using the enumeration $h$.
    \begin{figure}[h!]
    \centering
\begin{minipage}{0.32\textwidth}
\begin{center}
\includegraphics{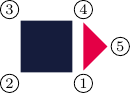}
\end{center}
\end{minipage}
\begin{minipage}{0.32\textwidth}
\begin{center}
\includegraphics{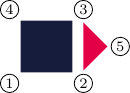}
\end{center}
\end{minipage}
\begin{minipage}{0.32\textwidth}
\begin{center}
\includegraphics{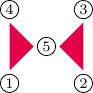}
\end{center}
\end{minipage}
\caption{Illustration of the three compiled hypergraphs in Example \ref{examp:comphypergraphs}. Left: $\PG_{B_1}$. Center: $\PG_{B_2}$. Right: $\PG_{B_3}$.}
\label{fig:comphypergraphs}
    \end{figure}
    Since $\PG_{B_1}$ and $\PG_{B_2}$ are isomorph, we see that $\PG_H=\set{[\PG_{B_1}],[\PG_{B_2}],[\PG_{B_3}]}=\set{[\PG_{B_1}],[\PG_{B_3}]}$.
\end{example}
\begin{remark}
    \label{rem:comphypergraphs}
    \begin{enumerate}
        \item The definition of the set of compiled hypergraphs $\PG_H$ of $H$ is well-defined, i.e., independent of the choice of enumeration $h$ for each basis $B\in\BCS_{\CS_H}$.
        \item \label{item:csedgeset} For all $B_1,B_2\in \BCS_{\CS_H}$, we have $E_{B_1}=E_{B_2}$.
        \begin{proof}
            Suppose that $e'\in E_{B_2}$. Therefore, there exists $C'\in B_2$ with $e'\in C'$. Since $C'\in\CS_H$ we can find $d_1,\ldots,d_n\in\GF$ such that \[C'=d_1 C_1\triangle\ldots\triangle d_nC_n\subset\bigcup_{i=1}^n C_i,\] where $n\coloneqq \dim(\CS_H)$ and $C_1,\ldots,C_n$ are the basis elements of $B_1$. This and the relation $\bigcup_{i=1}^n C_i=E_{B_1}$ show that $e'\in E_{B_1}$. Repeating the same calculations with $E_{B_1}$ instead of $E_{B_2}$ shows the desired equality.
        \end{proof}
    \end{enumerate}
\end{remark}
\noindent As a last step, we prove the following lemma, which is essential for the subsequent definition of the parity graph mapping. It ensures that the parity graph mapping is well-defined.
\begin{lemma}
    \label{lem:welldefinedparmap}
    Let $H,H'\in\HGS$ be two isomorphic hypergraphs. Then, the sets of compiled hypergraphs of $H$ and $H'$ are equal, i.e., we have $\PG_H=\PG_{H'}$.
\end{lemma}
\begin{proof}
    Let $H'=(V',E')\in\HGS$ be another representative of $[H]$ and let $f\colon V\to V'$ be a bijective function with property \eqref{eq:isomhypergraphs}, where $V$ denotes the vertex set and $E$ the edge set of the hypergraph $H$. Now, let us assume $B\in \BCS_{\CS_H}$ and an enumeration $h\colon\set{1,\ldots,m}\to E_B$. From Remark \ref{rem:cs} \ref{item:isomcs} we deduce \[B'\coloneqq\set{\set{f(e)\mid e\in C}\mid C\in B}\subset\CS_{H'}\setminus\set{\emptyset}\] and $B'\in \BCS_{\CS_{H'}}$. Next, using Remark \ref{rem:comphypergraphs} \ref{item:csedgeset} we define the enumeration $h'\colon\set{1,\ldots,m}\to E_{B'}\colon v\mapsto f(h(v))$ of $E_{B'}$ and observe for every $C\in B$
    \begin{equation*}
        h^{-1}(C)=\bigcup_{e\in C} h^{-1}(e)=\bigcup_{e\in C} (h')^{-1}(f(e))=(h')^{-1}(\set{f(e)\mid e\in C}),
    \end{equation*}
    where we used the relation $(h')^{-1}(e')=h^{-1}(f^{-1}(e'))$ for $e'\in E_{B'}$. Therefore, we can deduce for $e\subset\set{1,\ldots,m}$
    \begin{equation*}
        e\in E_{\PG_{B}}\Longleftrightarrow \exists C\in B\colon e=h^{-1}(C)\Longleftrightarrow \exists C\in B\colon e=(h')^{-1}(\set{f(e)\mid e\in C}),
    \end{equation*}
    which is equivalent to
    \begin{equation*}
        \exists C'\in B'\colon e=(h')^{-1}(C')\Longleftrightarrow e\in E_{\PG_{B'}}.
    \end{equation*}
    Thus, we have $\PG_B=\PG_{B'}$ and hence $\PG_H=\PG_{H'}$, which proves the desired statement.
\end{proof}
\begin{definition}[Parity graph mapping]
    Let $\HGSE\coloneqq\HGS/{\sim}=\set{[H]\mid H\in\HGS}$ and $\PGS\coloneqq\set{\PG_H\mid H\in\HGS}$ be the set containing every set of compiled hypergraphs. The \emph{parity graph mapping} is defined as \[\parmap \colon \HGSE\to\PGS\colon [H]\mapsto \PG_H.\]
\end{definition}
\section{Uniqueness of the parity graph mapping}
\label{sec:uniquenesssimplegraph}
In this section, we examine the parity graph mapping $\parmap$ in more detail and derive an equivalent condition for being injective on any subdomain of $\HGSE$. We will investigate the case, where the equivalence classes of the subdomain consist of graphs. Therefore, by Remark \ref{rem:cs} \ref{item:csequalls}, we may assume that their constraint space corresponds to their cycle space. 

\noindent The main theorem in this section states that there is a connection between injectivity and so-called loop labelings of compiled hypergraphs, which are defined in the following.
\begin{definition}
    For a hypergraph $H=(V,E)$ we call a function \[\lab\colon V\to \set{\set{i,j}\mid i,j\in\N,\ i\neq j}\] a \emph{loop labeling of $H$}, if it satisfies the following two properties:
    \begin{enumerate}
        \item $\lab$ is injective.
        \item All edges $e\in E$ satisfy $\lab(e)\in\LS_{G_{H,\lab}}$, where $G_{H,\lab}$ denotes the graph
        \begin{equation*}
            G_{H,\lab}\coloneqq(V_{G_{H,\lab}},E_{G_{H,\lab}})\quad\text{with}\quad V_{G_{H,\lab}}\coloneqq\bigcup_{e\in E}\bigcup_{v\in e}\lab(v),\quad E_{G_{H,\lab}}\coloneqq\bigcup_{e\in E}\bigcup_{v\in e}\set{\lab(v)}.
        \end{equation*}
    \end{enumerate}
    Moreover, we call $G_{H,\lab}$ the \emph{induced graph of the hypergraph $H$ with loop labeling $\lab$}.
\end{definition}
\begin{example}
    \label{examp:looplabelings}
    Let $H=(V,E)\in\HGS$ be the hypergraph with vertex set $V=\set{1,2,3,4,5,6,7}$ and edge set $E=\set{\set{1,2,3,4},\set{3,4,5,6,7}}$. Then the function $\lab_1\colon V\to\set{\set{i,j}\mid i,j\in\N,\ i\neq j}$ defined by
    \begin{center}
        \begin{tabular}{c | c | c | c | c | c | c | c}
            $v$ & $1$ & $2$ & $3$ & $4$ & $5$ & $6$ & $7$\\\hline
            $\lab_1(v)$ & $\set{1,2}$ & $\set{1,4}$ & $\set{2,3}$ & $\set{3,4}$ & $\set{2,6}$ & $\set{4,5}$ & $\set{5,6}$
        \end{tabular}
    \end{center}
    is a loop labeling of $H$ since $\lab_1$ is injective and $\lab_1(\set{1,2,3,4})=\set{\set{1,2},\set{1,4},\set{2,3},\set{3,4}}\in\LS_{G_{H,\lab_1}}$ as well as $\lab_1(\set{3,4,5,6,7})=\set{\set{2,3},\set{3,4},\set{2,6},\set{4,5},\set{5,6}}\in\LS_{G_{H,\lab_1}}$. Note that $\lab_2\colon V\to\set{\set{i,j}\mid i,j\in\N,\ i\neq j}$ set by
    \begin{center}
        \begin{tabular}{c | c | c | c | c | c | c | c}
            $v$ & $1$ & $2$ & $3$ & $4$ & $5$ & $6$ & $7$\\\hline
            $\lab_2(v)$ & $\set{1,2}$ & $\set{1,4}$ & $\set{2,3}$ & $\set{3,4}$ & $\set{1,5}$ & $\set{3,5}$ & $\set{2,4}$
        \end{tabular}
    \end{center}
    is also a loop labeling of $H$ but the induced graphs $G_{H,\lab_1}$ and $G_{H,\lab_2}$ are not isomorphic since the number of vertices of $G_{H,\lab_1}$ is six whereas of $G_{H,\lab_2}$ is five. Figure \ref{fig:looplabelings} illustrates the induced graphs $G_{H,\lab_1}$ and $G_{H,\lab_2}$ of $H$.
    \begin{figure}
        \begin{center}
            \includegraphics{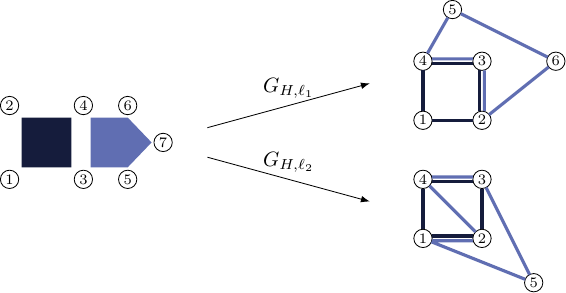}
        \end{center}
        \caption{Illustration of the induced graphs of the two labelings in Example \ref{examp:looplabelings}.}
        \label{fig:looplabelings}
    \end{figure}
\end{example}
\noindent The main theorem of this section reads as follows.
\begin{theorem}
    \label{thm:mainthm}
    Let $\DHGSE\coloneqq\set{[H]\in\HGSE\mid H\in\HGS\text{ is a graph}}$ be the set of all equivalence classes of graphs. Moreover, let $x\in\DHGSE$ and define $y\coloneqq\parmap(x)$. Then, we have 
    \begin{multline}
        \label{eq:eqvinj}
        \parmap\vert_\DHGSE^{-1}(\set{y})\\
        =\set{[H']\in\DHGSE\left\vert\,\exists\,\text{$\lab$ loop labeling of $\PG_{B}$}\colon \dim(\LS_{G_{\PG_{B},\lab}})=\abs{B}\text{ and $\LS_{G_{\PG_{B},\lab}}$ is cycle space of $H'$}\right.}
    \end{multline}
    for any $H\in x$ and any basis $B\in\BCS_{\LS_H}$.
\end{theorem}
\begin{remark}
    The expression of the preimage in \eqref{eq:eqvinj} does not depend on the representative $H$ of $x$ and the basis $B\in\BCS_{\LS_H}$.
\end{remark}
 \noindent As a consequence of statement Theorem \ref{thm:mainthm}, we obtain the following corollary, which presents an equivalent condition for the parity graph mapping being injective on any subdomain $\SD$ of the set of all equivalence classes of graphs $\DHGSE$.
 \begin{corollary}
    Let $\SD\subset\DHGSE$ be a subset of the set of all equivalence classes of graphs. The parity graph mapping $\parmap\colon \HGSE\to\PGS$ is injective onto the domain $\SD$ if and only if for all $x\in\SD$ there exists $H\in x$ and a basis $B\in\BCS_{\LS_H}$ such that 
    \begin{multline*}
        \abs{\set{[H']\in\SD\left\vert\,\exists\,\text{$\lab$ loop labeling of $\PG_{B}$}\colon \dim(\LS_{G_{\PG_{B},\lab}})=\abs{B}\text{ and $\LS_{G_{\PG_{B},\lab}}$ is cycle space of $H'$}\right.}}
        =1.
    \end{multline*}
 \end{corollary}
\noindent For the proof of Theorem \ref{thm:mainthm} we make use of the following three lemmas.
\begin{lemma}
    \label{lem:looplabisom}
    Let $H=(V,E)$ be a hypergraph with a loop labeling $\lab\colon V\to \set{\set{i,j}\mid i,j\in\N,\ i\neq j}$ and $H'=(V',E')$ an isomorphic hypergraph to $H$. Then, there exists a loop labeling for $H'$ for which the induced graph of $H'$ is equal to the induced graph of $H$ with loop labeling $\lab$.
\end{lemma}
\begin{proof}
    Since $H$ is isomorphic to $H'$, we can find a bijection $f\colon V\to V'$ satisfying \eqref{eq:isomhypergraphs}. Define the map $\lab'\coloneqq \lab\circ f^{-1}$, which is as a composition of injective functions an injection. Moreover, from \[y\in\im(\lab')\Longleftrightarrow\exists v'\in V'\colon y=\lab'(v')\Longleftrightarrow\exists v\in V\colon y=\lab(v)\] for $y\in\set{\set{i,j}\mid i,j\in\N,\ i\neq j}$, and the definitions of the induced graphs $G_{H,\lab}$ and $G_{H',\lab'}$, we also observe that $V_{G_{H,\lab}}=V_{G_{H',\lab'}}$ and $E_{G_{H,\lab}}=E_{G_{H',\lab'}}$, which proves that the above statement holds true. 
\end{proof}
\begin{lemma}
    \label{lem:existlooplab}
    Let $H=(V,E)$ be a graph and $B\in\BCS_{\LS_H}$ be a basis of the cycle space of $H$. Then, there exists a hypergraph $H'\in[H]$ and a loop labeling $\lab$ of the compiled hypergraph $\PG_B$ with some enumeration $h\colon\set{1,\ldots,m}\to E_B$ satisfying $\LS_{G_{\PG_B,\lab}}=\LS_{H'}$.
\end{lemma}
\begin{proof}
    In the following, let $f\colon V\to\set{1,\ldots,\abs{V}}$ be some bijective function and $H_f=(V_f,E_f)$ the isomorphic hypergraph of $H$ induced by $f$ as defined in Remark \ref{rem:isomhyp} \ref{item:isomhyp}. We will show that \[\lab\colon\set{1,\ldots,m}\to\set{\set{i,j}\mid i,j\in\N,\ i\neq j}\colon v\mapsto \tilde{f}(h(v)),\] where $m\coloneqq\abs{B}$ and
    \begin{equation*}
        \tilde{f}\colon\PS(V)\to\PS(\set{1,\ldots,\abs{V}})\colon M\mapsto f(M), 
    \end{equation*}
    is a loop labeling of $\PG_B$. From the bijectivity of $h$ and $\tilde{f}$ we see that $\lab$ is injective. Next, we will verify that $\LS_{G_{\PG_B,\lab}}=\LS_{H_f}$. First note that from Remark \ref{rem:ls} \ref{item:isomls} we can deduce $B'\coloneqq\set{\tilde{f}(C)\mid C\in B}=f_{\LS}(B)$ is a basis of $\LS_{H_f}$. We observe
    \begin{equation}
        \label{eq:escgrel}
        E_{G_{\PG_B,\lab}}=\bigcup_{e\in E_{\PG_B}}\tilde{f}(h(e))=\bigcup_{C\in B}\tilde{f}(h(h^{-1}(C)))=\bigcup_{C\in B}\tilde{f}(C)=\bigcup_{C'\in B'}C'.
    \end{equation}
    For any element $\tilde{C}\in\LS_{H_f}$ we can find coefficients $d_{C'}\in\Z_2$ for $C'\in B'$ such that \[\tilde{C}=\symdif\limits_{C'\in B'} d_{C'}C'\subset\bigcup_{C'\in B'}C'=E_{G_{\PG_B,\lab}}.\] This relation and Remark \ref{rem:ls} \ref{item:subsetls} show $\LS_{H_f}\subset\LS_{G_{\PG_B,\lab}}$. Moreover, from \eqref{eq:escgrel} we see that $E_{G_{\PG_B,\lab}}\subset E_{H_f}$ and therefore, we can deduce again from Remark \ref{rem:ls} \ref{item:subsetls} $\LS_{G_{\PG_B,\lab}}\subset\LS_{H_f}$ showing $\LS_{G_{\PG_B,\lab}}=\LS_{H_f}$. By definition of the compiled hypergraph $\PG_B$, we have $h(e)\in\LS_H$ for each $e\in E_{\PG_B}$ and therefore $\lab(e)=\tilde{f}(h(e))=f_{\LS}(h(e))\in\LS_{H_f}$, which finally proves the above statement.
\end{proof}
\begin{lemma}
    \label{lem:condeqcomphyp}
    Let $H,H'$ be two graphs, $B\in\BCS_{\LS_H}$ a basis of the cycle space of $H$ and $B'\in\BCS_{\LS_{H'}}$ a basis of the cycle space of $H'$, respectively. If the compiled hypergraphs $\PG_{B}$ and $\PG_{B'}$ are isomorph, then $\parmap([H])=\parmap([H'])$.
\end{lemma}
\begin{proof}
    Suppose that $f\colon V_{\PG_{B}}\to V_{\PG_{B'}}$ is a bijection between the vertex sets of $\PG_{B}$ and $\PG_{B'}$ such that $e\in E_{\PG_{B}}$ if and only if $f(e)\in E_{\PG_{B'}}$ for $e\subset V_{\PG_{B}}$. In the following, we denote by $C_1,\ldots,C_n\in\LS_H$ with $n\coloneqq\dim(\LS_H)$ the basis elements of $B$, i.e. $B=\set{C_1,\ldots,C_n}$, and $h\colon\set{1,\ldots,m}\to E_B$ and $h'\colon\set{1,\ldots,m}\to E_{B'}$ the enumerations of $E_B$ and $E_{B'}$ for some $m\in\N$ as given in Definition \ref{def:comphyps}. Then, by definition of $E_B$ and the compiled hypergraphs $\PG_{B}$ and $\PG_{B'}$ as well as the hypergraph isomorphism $f$, we have that $C_i'\coloneqq (h'\circ f\circ h^{-1})(C_i)$ are the basis elements of $B'$ for $i=1,\ldots,n$. Thus, we have
    \begin{equation}
        \label{eq:reledges}
        f(h^{-1}(C_i))=(h')^{-1}(C_i')\in E_{\PG_{B'}} \quad\text{for all } i=1,\ldots,n.
    \end{equation}
     Now, let $c^{(i)}\coloneqq(c_1^{(i)},\ldots,c_n^{(i)})\in\Z_2^n\setminus\set{0}$ be pairwise distinct vectors for $i=1,\ldots,n$ such that
    \begin{equation*}
        \tilde{B}\coloneqq\set{\symdif\limits_{j=1}^n c_j^{(1)} C_j,\ldots,\symdif\limits_{j=1}^n c_j^{(n)} C_j},
    \end{equation*}
    is a basis of $\LS_H$. Since $g\colon \LS_H\to\LS_{H'}\colon C\mapsto(h'\circ f\circ h^{-1})(C)$ is a isomorphism between the vector spaces $\LS_H$ and $\LS_{H'}$, we have that
    \begin{equation*}
        \tilde{B'}\coloneqq\set{g(C_1),\ldots,g(C_n)}=\set{\symdif\limits_{j=1}^n c_j^{(1)} C_j',\ldots,\symdif\limits_{j=1}^n c_j^{(n)} C_j'}
    \end{equation*}
    is a basis of $\LS_{H'}$. We will show that $\PG_{\tilde{B}}$ is an isomorphic compiled hypergraph to $\PG_{\tilde{B'}}$, which shows the final statement.\\
    From Remark \ref{rem:comphypergraphs} \ref{item:csedgeset} and the bijectivity of $h$ and $h'$ we first observe that $h$ is enumeration of $E_{\tilde{B}}$ as well as $h'$ for $E_{\tilde{B'}}$. Thus, we obtain for $e\subset V_{\PG_{\tilde{B}}}$ from the properties of bijective functions and \eqref{eq:reledges}
    \begin{align*}
        e\in E_{\PG_{\tilde{B}}}&\Longleftrightarrow \exists i\in\set{1,\ldots,n}\colon e=h^{-1}\left(\symdif\limits_{j=1}^n c_j^{(i)} C_j\right)\\
        &\Longleftrightarrow \exists i\in\set{1,\ldots,n}\colon e=\symdif\limits_{j=1}^n  h^{-1}(c_j^{(i)} C_j)\\
        &\Longleftrightarrow \exists i\in\set{1,\ldots,n}\colon f(e)=\symdif\limits_{j=1}^n  f(h^{-1}(c_j^{(i)} C_j))\\
         &\Longleftrightarrow \exists i\in\set{1,\ldots,n}\colon f(e)=\symdif\limits_{j=1}^n  (h'^{-1})(c_j^{(i)} C_j').
    \end{align*}
    Since $\symdif\limits_{j=1}^n  (h')^{-1}(c_j^{(i)} C_j')=(h')^{-1}\left(\symdif\limits_{j=1}^n c_j^{(i)} C_j'\right)$ for all $i=1,\ldots,n$, we have $e\in E_{\PG_{\tilde{B}}}$ if and only if $f(e)\in E_{\PG_{\tilde{B'}}}$. This finally shows $\PG_{\tilde{B}}\cong\PG_{\tilde{B'}}$.
\end{proof}
\begin{proof}[Proof of Theorem \ref{thm:mainthm}]
    For the sake of brevity, we write 
    \begin{multline*}
    D_{H,B}\coloneqq\\\set{[H']\in\DHGSE\left\vert\,\exists\,\text{$\lab$ loop labeling of $\PG_{B}$}\colon \dim(\LS_{G_{\PG_{B},\lab}})=\abs{B}\text{ and $\LS_{G_{\PG_{B},\lab}}$ is cycle space of $H'$}\right.},
    \end{multline*}
    denoting the set in \eqref{eq:eqvinj} for some $H\in x$ and a basis $B\in\BCS_{\LS_H}$.\\
    For better readability, we divide the proof into several parts.
    \begin{enumerate}[wide=\parindent,leftmargin=0pt,align=left]
        \item \label{item:firststepmainthm} First, we show for each loop labeling $\lab$ of $\PG_{B}$ with $\dim(\LS_{G_{\PG_{B},\lab}})=\abs{B}$ that \[B_{\lab}\coloneqq\set{\lab(e)\mid e\in E_{\PG_{B}}}\] is a basis of the cycle space $\LS_{G_{\PG_B,\lab}}$.\\
        Since $\lab$ is injective, we first observe that $\abs{B_{\lab}}=\abs{E_{\PG_B}}=\abs{B}$. Next, we show that $B_{\lab}$ is a linearly independent set in $\LS_{G_{\PG_{B},\lab}}$, which shows together with the assumption $\dim(\LS_{G_{\PG_{B},\lab}})=\abs{B}$ that $B_{\lab}$ is a basis of $\LS_{G_{\PG_{B},\lab}}$. Let $E_{\PG_{B}}=\set{e_1,\ldots,e_m}$ and assume that for $c_1,\ldots,c_{m}\in\GF$ \[c_1\lab(e_1)\symdif\ldots\symdif c_{m}\lab(e_{m})=\emptyset.\] Therefore, applying a left inverse of $\lab$ on the above equation yields \[c_1 e_1\symdif\ldots\symdif c_{m} e_{m}=\emptyset.\] Then, using the injectivity of the corresponding enumeration $h\colon\set{1,\ldots,m}\to E_B$, we deduce \[c_1 h(e_1)\symdif\ldots\symdif c_{m} h(e_{m})=\emptyset.\] By definition, $\set{h(e_1),\ldots,h(e_m)}=B$, and thus, we obtain $c_1=c_2=\ldots=c_m=0$, which proves the linear independence of the elements in the set $B_{\lab}$.
        \item Secondly, we will verify that for $x'\in D_{H,B}$, we have $\parmap(x')=\parmap(x)$. Assume that $x'\in D_{H,B}$ and let $H'\in x'$ be a graph such that $\dim(\LS_{G_{\PG_{B},\lab}})=\abs{B}$ and $\LS_{G_{\PG_{B},\lab}}$ is the cycle space of $H'$ for some loop labeling $\lab$. From \ref{item:firststepmainthm} we have that $B_{\lab}$ is a basis of $\LS_{H'}$ and from
        \begin{equation*}
            E_{B_{\lab}}=\bigcup_{\tilde{C}\in B_{\lab}}\bigcup_{\tilde{e}\in\tilde{C}}\set{\tilde{e}}=\bigcup_{e\in E_{\PG_B}}\bigcup_{\tilde{e}\in\lab(e)}\set{\tilde{e}}=\bigcup_{e\in E_{\PG_B}}\lab(e)=E_{G_{\PG_B,\lab}}
        \end{equation*}
        we see that \[h'\colon V_{\PG_B}\to E_{B_{\lab}}\colon v\mapsto \lab(v)\] is well-defined and an enumeration of $E_{B_{\lab}}$. Moreover
        \begin{equation*}
            E_{\PG_{B_\lab}}=\bigcup_{\tilde{C}\in B_{\lab}}\set{(h')^{-1}(\tilde{C})}=\bigcup_{e\in E_{\PG_B}}\set{(h')^{-1}(\lab(e))}=\bigcup_{e\in E_{\PG_B}}\set{e}=E_{\PG_B},
        \end{equation*}
        which shows $\PG_{B_\lab}=\PG_B$. Therefore, from Lemma \ref{lem:condeqcomphyp} we can conclude that $D_{H,B}\subset \parmap\vert_\DHGSE^{-1}(\set{y})$.
        \item \label{item:secondstepmainthm1} In order to prove the other subset relation, we next show that for any two graphs $H_1,H_2$ with $[H_1],[H_2]\in\DHGSE$ and $\PG_{H_1}=\PG_{H_2}$, we have $D_{H_1,B_1}=D_{H_2,B_2}$ for all $B_1\in\BCS_{\LS_{H_1}}, B_2\in\BCS_{\LS_{H_2}}$. Suppose now $B_1\in\BCS_{\LS_{H_1}}$ and let $B_2\in\BCS_{\LS_{H_2}}$ such that $[\PG_{B_1}]=[\PG_{B_2}]$. Moreover, let $\lab_1$ be a loop labeling of $\PG_{B_1}$ such that $\LS_{G_{\PG_{B_1},\lab_1}}$ is the cycle space of some graph $H_1'=(V_1',E_1')$ with $\dim(\LS_{G_{\PG_{B_1},\lab_1}})=\abs{B_1}$. Since $\PG_{B_1}$ is an isomorphic hypergraph to $\PG_{B_2}$ we can find by Lemma \ref{lem:looplabisom} a loop labeling $\lab_2$ of $\PG_{B_2}$ such that the induced graph $G_{\PG_{B_1},\lab_1}=(V_{G_{\PG_{B_1},\lab_1}},E_{G_{\PG_{B_1},\lab_1}})$ is equal to the induced graph $G_{\PG_{B_2},\lab_2}=(V_{G_{\PG_{B_2},\lab_2}},E_{G_{\PG_{B_2},\lab_2}})$ with loop labeling $\lab_2$. Therefore, we can conclude that \[\dim(\LS_{G_{\PG_{B_2},\lab_2}})=\dim(\LS_{G_{\PG_{B_1},\lab_1}})=\abs{B_1}=\abs{E_{\PG_{B_1}}}=\abs{E_{\PG_{B_2}}}=\abs{B_2}\] and thus $[H_1']\in D_{H_2,B_2}$.  Analogously, we can deduce that $[H_2']\in D_{H_1,B_1}$ for an arbitrary element $[H_2']\in D_{H_2,B_2}$.
        \item As a final step, let $x'\in\parmap\vert_\DHGSE^{-1}(\set{y})$, $H'\in x$ and $B'\in\BCS_{\LS_{H'}}$. Since $\PG_H=\parmap(x)=\parmap(x')=\PG_{H'}$, we deduce from \ref{item:secondstepmainthm1} the equality $D_{H',B'}=D_{H,B}$. Furthermore, from Lemma \ref{lem:existlooplab} we have that $x'\in D_{H',B'}$, which finally shows $x'\in D_{H,B}$.\qedhere
    \end{enumerate}
 \end{proof}
 \noindent Finally, we will see that for certain cycle bases it is sufficient to consider only loop labelings $\lab$ appearing in the expression of the preimage \eqref{eq:eqvinj} which satisfy $\lab(e)\in \LSP_{G_{\PG_H,\lab}}$ for all $e\in E_{\PG_H}$. The statement reads as follows.
 \begin{corollary}
    \label{cor:preimagefundbasis}
    Let $x\in\DHGSE$ and define $y\coloneqq\parmap(x)$. Moreover, let $H\in x$ and $B\in\BCS_{\LS_H}$ where $B\subset \LSP_{H}$ is a weakly fundamental basis of $\LS_H$, i. e., there exists an enumeration $B=\set{C_1,\ldots,C_n}$ with
    \begin{equation}
        \label{eq:weaklyfundbasis}
        \forall 2\leq k\leq n\colon C_k\setminus\bigcup_{l=1}^{k-1} C_l\neq \emptyset    
    \end{equation}
    and $n\coloneqq \abs{B}$. Then, we have 
    \begin{align*}
        \parmap\vert_\DHGSE^{-1}(\set{y})=\big\{&[H']\in\DHGSE\,\big\vert\,\exists\,\text{$\lab$ loop labeling of $\PG_{B}$ with $\lab(e)\in \LSP_{G_{\PG_{B},\lab}}$ for all $e\in E_{\PG_{B}}$},\\
        &\dim(\LS_{G_{\PG_{B},\lab}})=n\text{ and }\LS_{G_{\PG_{B},\lab}}=\LS_{H'}\big\}.
    \end{align*}
\end{corollary}
\noindent Besides Theorem \ref{thm:mainthm}, the proof of this corollary is based on the following two lemmas, whereas the first one will be also used for our derivations in Section \ref{sec:applications}.
 \begin{lemma}
    \label{lem:restrlooplab}
    Let $H=(V,E)$ be a hypergraph, where the edge set $E=\set{e_1,\ldots,e_n}$ satisfies
    \begin{equation}
        \label{eq:condhypergraph}
        \forall 2\leq k\leq n\colon e_k\setminus\bigcup_{l=1}^{k-1} e_l\neq \emptyset.
    \end{equation}
    Moreover, let $\lab$ be a loop labeling of $H$.
    \begin{enumerate}
        \item \label{item:linearind} The set $B\coloneqq\set{\lab(e_1),\ldots,\lab(e_n)}$ is linearly independent in $\LS_{G_{H,\lab}}$.
        \item \label{item:restrlooplab} If $n\geq 2$ and $\dim(\LS_{G_{H,\lab}})=n$, then for all $1\leq k\leq n$ the function $\lab_{k}\coloneqq \lab\vert_{V_{k}}$ is a loop labeling of the hypergraph $H_{k}\coloneqq (V_{k},E_{k})$ with $V_{k}\coloneqq\set{v\in V\mid \exists 1\leq l\leq k\colon v\in e_l}$ and $E_{k}\coloneqq\bigcup_{l=1}^{k}\set{e_l}$ satisfying $\dim(\LS_{G_{H_{k},\lab_{k}}})=k$.
    \end{enumerate}
\end{lemma}
\begin{proof}
    \begin{enumerate}[wide=\parindent,leftmargin=0pt,align=left]
        \item Since $\lab(e_1)\neq\emptyset$, we first see that $\set{\lab(e_1)}$ is linearly independent. Now, let $1\leq m< n$ and assume that $\set{\lab(e_1),\ldots,\lab(e_m)}$ is linearly independent. Moreover, let $c_1,\ldots,c_{m+1}\in\GF$ with \[c_1\lab(e_1)\symdif\ldots\symdif c_{m+1}\lab(e_{m+1})=\emptyset.\] Thus, we have \[c_1\lab(e_1)\symdif\ldots\symdif c_{m}\lab(e_{m})=c_{m+1}\lab(e_{m+1})\] and from the relation \[\lab(e_{m+1})\setminus\bigcup\limits_{k=1}^m c_k \lab(e_k)\neq \emptyset,\] we deduce $c_{m+1}=0$. Therefore, we have $c_1\lab(e_1)\symdif\ldots\symdif c_{m}\lab(e_{m})=\emptyset$ and hence, from our assumption, $c_1=\ldots=c_m=0$, which proves the linear independence of the elements in the set $B$.
        \item By definition of $\lab$, we have that $\lab_k$ is a loop labeling with $\lab_k(e)\in \LS_{G_{H_k,\lab_k}}$ for all $e\in E_k$. Next, we show the statement for $k=n-1$. Assume by contradiction that $\dim(\LS_{G_{H_{n-1},\lab_{n-1}}})\geq n$. Since $e_{n}\setminus\bigcup_{k=1}^{n-1} e_k\neq \emptyset$, we have $\dim(\LS_{G_{H,\lab}})\geq n+1$, which is a contradiction. Thus, $\dim(\LS_{G_{H_{n-1},\lab_{n-1}}})\leq n-1$ and by \ref{item:linearind} we have $\dim(\LS_{G_{H_{n-1},\lab_{n-1}}})=n-1$. The remaining statement follows from an inductive argument.\qedhere
    \end{enumerate}
\end{proof}
 \begin{lemma}
    \label{lem:assumptedgesetspanbasis}
    Let $H=(V,E)$ be a hypergraph, where the edge set $E=\set{e_1,\ldots,e_n}$ satisfies
    \begin{equation}
        \label{eq:assumptedgesetspanbasis}
        \forall 2\leq k\leq n\colon\left(e_k\setminus\bigcup_{l=1}^{k-1} e_l\neq \emptyset\quad\text{and}\quad \forall e\in\mathrm{span}_{\GF}(\set{e_1,\ldots,e_{k-1}})\setminus\set{\emptyset}\subset\PS(V)\colon e\setminus e_k\neq\emptyset\right).
    \end{equation}
    Then, for every loop labeling $\lab$ of $H$ with $\dim(\LS_{G_{H,\lab}})=n$ we have $\lab(e_i)\in \LSP_{G_{H,\lab}}$ for all $i=1,\ldots,n$.
\end{lemma}
\begin{proof}
    We show the statement by verifying it for $H_k=(V_k,E_k)$, where we used the notation in Lemma \ref{lem:restrlooplab} \ref{item:restrlooplab} for $k=1,\ldots,n$. Suppose for $k=1$ we have $\lab_1(e_1)=\lab(e_1)\not\in \LSP_{G_{H,\lab}}$. Therefore, we can find disjoint sets $\tilde{e}_1,\tilde{e}_2\subset e_1$ with $\lab_1(\tilde{e}_l)\in \LSP_{G_{H,\lab}}$ for $l=1,2$. Thus, we observe that $\set{\lab_1(\tilde{e}_1),\lab_1(\tilde{e}_2)}$ is linearly independent in $\LS_{G_{H_1,\lab_1}}$ and hence $\dim(\LS_{G_{H_1,\lab_1}})\geq 2$, which is a contradiction by Lemma \ref{lem:restrlooplab} \ref{item:restrlooplab}. Now, assume that statement holds for some $1\leq k<n$ and suppose that there exist $\tilde{e}_1,\tilde{e}_2\subset e_{k+1}$ with $\lab(\tilde{e}_l)\in \LSP_{G_{H,\lab}}$ for $l=1,2$. If $\lab(\tilde{e}_l)\in \LS_{G_{H_{k},\lab_{k}}}$ for some $l\in\set{1,2}$, then we could find $c_1,\ldots,c_{k}\in\GF$ with $\lab(\tilde{e}_l)=c_1\lab(e_1)\triangle\ldots\triangle c_{k}\lab(e_k)$. By injectivity of $\lab$, we could derive $\tilde{e}_l=c_1 e_1\triangle\ldots\triangle c_{k}e_{k}$, which is a contradiction to the second assumption in \eqref{eq:assumptedgesetspanbasis}. This implies that $\lab(\tilde{e}_l)\not\in \LS_{G_{H_{k},\lab_{k}}}$ for $l=1,2$ showing that $\set{\lab_{k+1}(e_1),\ldots,\lab_{k+1}(e_{k}),\lab_{k+1}(\tilde{e}_1),\lab_{k+1}(\tilde{e}_2)}$ is a linearly independent set in $\LS_{G_{H_{k+1},\lab_{k+1}}}$ and hence, $\dim(\LS_{G_{H_{k+1},\lab_{k+1}}})\geq k+2$, which cannot be true by Lemma \ref{lem:restrlooplab} \ref{item:restrlooplab}.\qedhere
\end{proof}
\begin{proof}[Proof of Corollary \ref{cor:preimagefundbasis}]
    This follows from the fact that cycles cannot contain other cycles, which implies that the compiled hypergraph $\PG_B$ satisfies assumption \eqref{eq:assumptedgesetspanbasis} in Lemma \ref{lem:assumptedgesetspanbasis} by using the bijectivity of an enumeration $h\colon\set{1,\ldots,m}\to E_B$ given in Definition \ref{def:comphyps}. Applying Theorem \ref{thm:mainthm} yields then the claimed identity.\qedhere
\end{proof}
 \section{Extensions and applications}
 \label{sec:applications}

 \noindent In the last section of this article, we apply the derived results in Section \ref{sec:uniquenesssimplegraph} to investigate further properties of the parity graph mapping on the set of equivalence classes of graphs. Specifically, in Section \ref{subsec:basic_cond_non_uniqueness}, we outline basic conditions that result in non-uniqueness under the parity graph mapping. In Section \ref{subsec:uniqueness_one_two_dim}, we demonstrate that graphs whose edges lie in their cycle space, and for which the constraint space has dimension at most two, are mapped to different compiled hypergraphs. In contrast, Section \ref{subsec:counterexamples} shows that the statement does not hold for graphs with constraint space of dimension greater than two. Finally, Section \ref{subsec:charrectplaqlayouts} characterizes the preimage of rectangular plaquette layouts under the parity graph mapping and introduces a polynomial-time algorithm that can determine whether a given optimization problem can be compiled onto such a plaquette layout, while also generating the corresponding physical layout. This is particularly relevant for certain quantum hardware implementations that utilize square grid-like arrangements.
 
\noindent In the following, for a graph $H=(V,E)\in\HGS$ and $E'\subset E$ we denote by $H\vert_{E'}\coloneqq (V',E')$ the restricted graph of $H$ on $E'$, where $V'\coloneqq\set{v\in V\mid \exists e'\in E'\colon v\in e}$. Moreover, we use $E_{\LS}\coloneqq\set{e\in E\mid \exists C\in\LS_H\colon e\in C}$ to denote the set of edges which are contained in an element of the cycle space $\LS_H$.
 \subsection{Basic conditions for non-uniqueness}
 \label{subsec:basic_cond_non_uniqueness}
 \noindent In this section, we present core criteria that violate the one-to-one correspondence of the parity graph mapping. We begin with the following example, which demonstrates that, in general, for every element $x\in\DHGSE$ the preimage $\parmap\vert_{\DHGSE}^{-1}(\set{y})$ with $y\coloneqq\parmap(x)$ is not a singleton. 
\begin{example}
    Let $H_1=(V_1,E_1)$ be the graph as defined in Example \ref{examp:comphypergraphs} and $H_2=(V_2,E_2)$ be the graph defined by \[V_2\coloneqq\set{1,2,3,4,5}\quad\text{and}\quad E_2\coloneqq\set{\set{1,2},\set{2,3},\set{3,4},\set{1,4},\set{2,4},\set{4,5}},\] which is not isomorphic to $H_1$. Using the same cycle basis $$B_1=\set{\set{\set{1,2},\set{2,3},\set{3,4},\set{4,1}},\set{\set{1,2},\set{2,4},\set{4,1}}}$$ as in Example \ref{examp:comphypergraphs} and the enumeration
    \begin{center}
        \begin{tabular}{c | c | c | c | c | c | c}
            $i$ & $1$ & $2$ & $3$ & $4$ & $5$ & $6$\\\hline
            $h(i)$ & $\set{1,2}$ & $\set{2,3}$ & $\set{3,4}$ & $\set{4,1}$ & $\set{2,4}$ & $\set{4,5}$
        \end{tabular}
    \end{center}
    we observe that compiled hypergraphs of $H_1$ and $H_2$ are equal. Thus, by Lemma \ref{lem:condeqcomphyp}, we have $\parmap([H_1])=\parmap([H_2])$.
\end{example}
\noindent As can be observed from this example, adding a new edge and a new vertex to a existing graph $H=(V,E)$ which does not form any new cycle results in the same cycle space and therefore yielding the same set of compiled hypergraphs. However, the next example demonstrates that even for two non-isomorphic hypergraphs $H=(V,E)$ and $H'=(V',E')$ which satisfy $H=H\vert_{E_{\LS}}$ and $H'=H'\vert_{E'_{\LS}}$ the equality $\parmap([H])=\parmap([H'])$ can hold.
\begin{example}
    \label{examp:eqcomphypergraph}
    Let $H_1=(V_1,E_1), H_2=(V_2,E_2)$ be two graphs defined by $V_1\coloneqq V_2 \coloneqq\set{1,2,3,4,5,6}$,
    \begin{align*}
        E_1\coloneqq\set{\set{1,2},\set{2,3},\set{3,4},\set{1,4},\set{2,4},\set{3,5},\set{5,6},\set{3,6}}\quad\text{and}\\
        E_2\coloneqq\set{\set{1,2},\set{2,3},\set{3,4},\set{1,4},\set{2,4},\set{4,5},\set{5,6},\set{4,6}}.
    \end{align*}
    We see that the degree of vertex $4$ in $H_2$ is five whereas the degree of every vertex in $H_1$ is less than five and hence, $H_1$ is not an isomorphic graph to $H_2$. However, using the bases
    \begin{align*}
        B_1\coloneqq\set{\set{\set{1,2},\set{2,4},\set{1,4}},\set{\set{2,3},\set{3,4},\set{2,4}},\set{\set{3,5},\set{5,6},\set{3,6}}}\quad\text{and}\\
        B_2\coloneqq\set{\set{\set{1,2},\set{2,4},\set{1,4}},\set{\set{2,3},\set{3,4},\set{2,4}},\set{\set{4,5},\set{5,6},\set{4,6}}}
    \end{align*}
    as well as the enumerations
    \begin{center}
        \begin{tabular}{c | c | c | c | c | c | c | c | c}
            $i$ & $1$ & $2$ & $3$ & $4$ & $5$ & $6$ & $7$ & $8$\\\hline
            $h_1(i)$ & $\set{1,2}$ & $\set{1,4}$ & $\set{2,4}$ & $\set{2,3}$ & $\set{3,4}$ & $\set{3,5}$ & $\set{5,6}$ & $\set{3,6}$
        \end{tabular}
    \end{center}
    \begin{center}
        \begin{tabular}{c | c | c | c | c | c | c | c | c}
            $i$ & $1$ & $2$ & $3$ & $4$ & $5$ & $6$ & $7$ & $8$\\\hline
            $h_2(i)$ & $\set{1,2}$ & $\set{1,4}$ & $\set{2,4}$ & $\set{2,3}$ & $\set{3,4}$ & $\set{4,5}$ & $\set{5,6}$ & $\set{4,6}$
        \end{tabular}
    \end{center}
    we have that $\PG_{B_1}=\PG_{B_2}$. Therefore, we again obtain from Lemma \ref{lem:condeqcomphyp} that $\parmap([H_1])=\parmap([H_2])$. Figures \ref{fig:example_graph_1}, \ref{fig:example_graph_2} illustrate the graphs $H_1$ and $H_2$ and Figure \ref{fig:compiled_hypergraph_12} the compiled hypergraph $\PG_{B_1}=\PG_{B_2}$.
\end{example}
\begin{figure}[h!]
    \begin{minipage}{0.32\textwidth}
\begin{center}
    \includegraphics{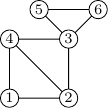}
    \captionof{figure}{Illustration of the graph $H_1$ in Example \ref{examp:eqcomphypergraph}.}
    \label{fig:example_graph_1}
\end{center}
\end{minipage}
\hfill
    \begin{minipage}{0.32\textwidth}
\begin{center}
    \includegraphics{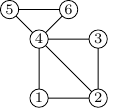}
    \captionof{figure}{Illustration of the graph $H_2$ in Example \ref{examp:eqcomphypergraph}.}
    \label{fig:example_graph_2}
\end{center}
\end{minipage}
    \begin{minipage}{0.32\textwidth}
\begin{center}
    \includegraphics{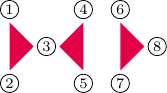}
    \captionof{figure}{Compiled hypergraph $\PG_{B_1}=\PG_{B_2}$ in Example \ref{examp:eqcomphypergraph}.}
    \label{fig:compiled_hypergraph_12}
\end{center}
\end{minipage}
\hfill
\end{figure}
\noindent The reason for the above equality is that in both cases one cycle is disconnected from the other two cycles and hence leading to compiled hypergraphs which have one disconnected edge. For that reason, we focus in the subsequent considerations on graphs whose cycle space is \emph{connected}. Below, we give a precise definition of a connected cycle space.
  \begin{definition}[Connected cycle space]
     Let $H=(V,E)$ be a hypergraph and $\LSP_H$ the spanning set of its cycle space and $M\subset \LSP_H$. We say that two elements $C,D\in \LSP_H$ are connected in $M$ if and only if there exist $C_1,\ldots,C_n\in M$ with $C_1=C$ and $C_n=D$ satisfying $C_{k-1}\cap C_{k}\neq\emptyset$ for all $k=2,\ldots,n$. Moreover, the relation
     \begin{equation*}
        C\sim_{M} D:\Longleftrightarrow C\text{ is connected to } D
    \end{equation*}
    for $C,D\in M$ defines an equivalence relation on $M$. We call each equivalence class $[C]$ of $M/{\sim_{M}}$ a \emph{connected component} of $M$. If the number of connected components $\abs{M/{\sim_{M}}}$ is equal to one, we say that \emph{$M$ is connected}. Furthermore, we call the cycle space $\LS_H$ connected if $\LSP_H$ is connected.
 \end{definition}
 \noindent Taking both observations into account, we continue with the next subsection, where we investigate graphs with connected cycle spaces of dimension one or two in more detail.
 \subsection{Uniqueness of compiled hypergraphs in one or two dimensions}
 \label{subsec:uniqueness_one_two_dim}
 The main aim of this subsection is to show that if two graphs $H=(V,E), H'=(V',E')\in\HGS$ have connected cycle spaces of dimension one or two and share the same set of compiled hypergraphs, then the restrictions $H\vert_{E_{\LS}}$ and $H'\vert_{E'_{\LS}}$ are isomorphic.
 \begin{theorem}
    \label{thm:uniquenessdimtwo}
    Let $H=(V,E),H'=(V',E')$ be two graphs whose cycle space is connected and assume $\dim(\LS_{H})=\dim(\LS_{H'})=n$ with $n\in\set{1,2}$. If $\parmap([H])=\parmap([H'])$, then $H\vert_{{E}_{\LS}}$ is isomorphic to $H'\vert_{E_{\LS}'}$.
\end{theorem}
\noindent The theorem is based on the following proposition.
\begin{prop}
    \label{prop:isomsimplegraphs12}
    Let $H=(V,E)$ be a hypergraph, where the edge set $E$ satisfies \eqref{eq:condhypergraph} and suppose $e_k\cap \bigcup_{l=1}^{k-1}e_l\neq\emptyset$ for all $2\leq k\leq n$. Moreover, let $\lab_1,\lab_2$ be two loop labelings of $H$ with $\dim(G_{H,\lab_1})=\dim(G_{H,\lab_2})=n$ and assume $\lab_1(e)\in \LSP_{G_{H,\lab_1}},\lab_2(e)\in \LSP_{G_{H,\lab_2}}$ for all $e\in E$. If $n\in\set{1,2}$, then $G_{H,\lab_1}$ is an isomorphic graph to $G_{H,\lab_2}$.
\end{prop}
\noindent In order to prove Proposition \ref{prop:isomsimplegraphs12}, we make use of the following lemma, which is one of the key ingredients for our derivations in this section. It states that loop labelings of certain hypergraphs must satisfy additional conditions which are important for characterizing the preimage in \eqref{eq:weaklyfundbasis}.
\begin{lemma}
    \label{lem:outerloops}
    Let $H=(V,E)$ be a hypergraph, where the edge set $E$ satisfies \eqref{eq:condhypergraph} and suppose ${e_k\cap \bigcup_{l=1}^{k-1}e_l\neq\emptyset}$ for all $2\leq k\leq n$. Moreover, let $\lab$ be a loop labeling of $H$ with $\dim(G_{H,\lab})=n$ and $\lab(e)\in \LSP_{G_{H,\lab}}$ for all $e\in E$, and define
    \begin{equation}
        V_{e_n,\lab}\coloneqq\set{v\in V_{G_{H_{n-1},\lab_{n-1}}}\left\vert\,\exists e'\in E_{n-1}\colon \exists v_0\in e'\colon\exists v_1\in\tilde{e}_n\colon v\in\lab(v_0)\cap\lab(v_1)\right.},
    \end{equation}
    where we used the same notations as in Lemma \ref{lem:restrlooplab} \ref{item:restrlooplab} and $\tilde{e}_n\coloneqq e_n\setminus\bigcup_{i=1}^{n-1}e_i$.
    \begin{enumerate}
        \item \label{item:outerpath} There exists an enumeration of $\tilde{e}_n=\set{v_1,\ldots,v_m}$ with $m\coloneqq\abs{\tilde{e}_n}$ such that $\set{\lab(v_1),\ldots,\lab(v_m)}$ is a path in $G_{H,\lab}$ and $\lab(v_1)$ as well as $\lab(v_m)$ contain an element in $V_{e_n,\lab}$. Furthermore, for $m\geq 3$ and $k=2,\ldots,m-1$, we have $\lab(v_k)\cap V_{e_n,\lab}=\emptyset$.
        \item \label{item:innerpath} The set $\lab(e_n\setminus\tilde{e}_n)$ is a path in $G_{H_{n-1},\lab_{n-1}}$.
        \item \label{item:outerverticesset} $\abs{V_{e_n,\lab}}=2$.
    \end{enumerate}
\end{lemma}
\begin{proof}
    \begin{enumerate}[wide=\parindent,leftmargin=0pt,align=left]
        \item First, we show that $\abs{V_{e_n,\lab}}\geq 1$. Assume by contradiction that $V_{e_n,\lab}=\emptyset$ and thus, for all $v\in V_{G_{H_{n-1},\lab_{n-1}}}$ we have that \[\forall e'\in E_{n-1}\colon \forall v_0\in e'\colon\forall v_1\in\tilde{e}_n\colon v\not\in\lab(v_0)\cap\lab(v_1).\] This implies that for all $v_0\in e_n\setminus\tilde{e}_n$ and $v_1\in\tilde{e}_n$ we have $\lab(v_0)\cap\lab(v_1)=\emptyset$ and therefore, there exist no walks between vertices in $\lab(v_0)$ and $\lab(v_1)$ for any $v_0\in e_n\setminus\tilde{e}_n$ and $v_1\in\tilde{e}_n$, which cannot be true since $\lab(e_n)$ is a cycle. This also shows that the statement is true for $m=1$. Now, we show that for all $1\leq k\leq m$ there exists pairwise distinct $v_1,\ldots,v_k\in\tilde{e}_n$ such that $\set{\lab(v_1),\ldots,\lab(v_k)}$ is a path and $\lab(v_1)$ contains an element in $V_{e_n,\lab}$. From the previous argument, we have that the statement holds for $k=1$. Now, suppose the statement holds for some $1\leq k<m$. Therefore, we can find $v_1,\ldots,v_k\in\tilde{e}_n$ and $i_1,\ldots,i_{k+1}\in\N$ with $\lab(v_l)=\set{i_l,i_{l+1}}$ for $l=1,\ldots,k$ such that $\set{\lab(v_1),\ldots,\lab(v_k)}$ is a path with $\lab(v_1)\cap V_{e_n,\lab}\neq\emptyset$. Since $\lab(e_n)$ is a cycle, we can find $v_{k+1}\in e_{n}\setminus\set{v_1,\ldots,v_k}$ with $\lab(v_k)\cap\lab(v_{k+1})\neq\emptyset$. %and $i_{k+2}\in\N$ such that $\lab(v_{k+1})=\set{i_{k+1},i_{k+2}}$
        If $v_{k+1}\not\in\tilde{e}_n$, then we see that $\lab(v_{k})\cap V_{G_{H_{n-1},\lab_{n-1}}}\neq\emptyset$ and therefore, we can find a path $p=\set{\set{i_{k+1},j_1},\ldots,\set{j_l,i_{1}}}$ with $j_1,\ldots,j_l\in\N$ in $V_{G_{H_{n-1},\lab_{n-1}}}$ since $G_{H_{n-1},\lab_{n-1}}$ is connected. Thus, we have $C\coloneqq\set{\lab(v_1),\ldots,\lab(v_k)}\cup p\in \LS_{G_{H,\lab}}$. Since $C\setminus\bigcup_{l=1}^{n-1}\lab(e_l)\neq\emptyset$ and $\lab(e_n)\setminus(C\cup\bigcup_{l=1}^{n-1}\lab(e_l))\neq\emptyset$, we observe that $\set{\lab(e_1),\ldots,\lab(e_{n-1}),C,\lab(e_n)}$ is a linearly independent set in $\LS_{G_{H,\lab}}$, which is a contradiction to $\dim(\LS_{G_{H,\lab}})=n$. Using the same arguments, it can be shown for $m\geq 3$ that $\lab(v_k)\cap V_{e_n,\lab}=\emptyset$ for $k=2,\ldots,m-1$. Moreover, since $\lab(e_n)$ is a cycle, we can find $w\in e_n\setminus \tilde{e}_n\subset\bigcup_{k=1}^{n-1}e_k$ with $\lab(v_m)\cap\lab(w)\neq\emptyset$ showing that $\lab(v_m)\cap V_{e_n,\lab}\neq\emptyset$.
        \item By using that $\lab(e_n)\in \LSP_H$ we can deduce from the first statement \ref{item:outerpath} that there exist vertices $v_1,\ldots,v_m\in\tilde{e}_n$, $v_{m+1},\ldots,v_{\abs{e_n}}\in e_n\setminus\tilde{e}_n$ and $w_0,\ldots,w_{\abs{e_n}}\in V_{G_{H,\lab}}$ such that the sequence $(w_0,\lab(v_1),w_1,\lab(v_2),\ldots,\lab(v_{\abs{e_n}}),w_{\abs{e_n}})$ is a cycle in $G_{H,\lab}$. Hence, we immediately obtain that $\set{\lab(v_{m+1}),\lab(v_{m+2}),\ldots,\lab(v_{\abs{e_n}})}$ is a path in $G_{H,\lab}$.
        \item This follows immediately from \ref{item:outerpath}.\qedhere
    \end{enumerate}
\end{proof}
\begin{proof}[Proof of Proposition \ref{prop:isomsimplegraphs12}]
        For the case $E=\set{e_1}$ we can deduce from the assumption on the loop labelings $\lab_1$ and $\lab_2$ that there exist pairwise distinct $i_1,\ldots,i_{m}\in\N$ and $j_1,\ldots,j_{m}\in\N$ with 
        \begin{equation}
            \label{eq:looplabelingdimone}
            \lab_1(e_1)=\set{\set{i_1,i_2},\set{i_2,i_3},\ldots,\set{i_{m},i_1}}\quad\text{and}\quad \lab_2(e_1)=\set{\set{j_1,j_2},\set{j_2,j_3},\ldots,\set{j_{m},j_1}},
        \end{equation}
        where $m\coloneqq\abs{e_1}$. Defining the bijective function $f\colon V_{G_{H,\lab_1}}\to V_{G_{H,\lab_2}}$ by $f(i_k)\coloneqq j_k$ for $k=1,\ldots,m$, we observe that $f$ satisfies \eqref{eq:isomhypergraphs} and therefore $G_{H,\lab_1}$ is isomorphic to $G_{H,\lab_2}$.\\
        For $E=\set{e_1,e_2}$ and using the same notations as in \eqref{eq:looplabelingdimone}, we can find by Lemma \ref{lem:outerloops} \ref{item:outerpath} $k_1,k_2\in\set{1,\ldots,\abs{e_1}}$ such that $V_{e_2,\lab_1}=\set{i_{k_1},i_{(k_1+p-1\mod m) + 1}}$ and $V_{e_2,\lab_2}=\set{j_{k_2},j_{(k_2+p-1\mod m) + 1}}$, where $\mod$ denotes the modulo operator and $p\coloneqq\abs{e_1\cap e_2}$. Moreover, by Lemma  \ref{lem:outerloops} \ref{item:outerpath} and \ref{item:innerpath} there exist pairwise distinct numbers $i_1',\ldots,i_{\abs{e_2}-p-1}'\in\N\setminus\set{i_1,\ldots,i_m}$ and $j_1',\ldots,j_{\abs{e_2}-p-1}'\in\N\setminus\set{j_1,\ldots,j_m}$ with
        \begin{align*}
            \lab_1(e_2\setminus e_1)&=\set{\set{i_{(k_1+p-1\mod m) + 1},i_1'},\set{i_1',i_2'},\ldots,\set{i_{\abs{e_2}-p-1}',i_{k_1}}},\\
            \lab_1(e_1\cap e_2)&=\set{\set{i_{k_1},i_{(k_1 \mod m) + 1}},\ldots,\set{i_{(k_1+p-2\mod m) + 1},i_{(k_1+p-1\mod m) + 1}}},
        \end{align*}
        and
        \begin{align*}
            \lab_2(e_2\setminus e_1)&=\set{\set{j_{(k_2+p-1\mod m) + 1},j_1'},\set{j_1',j_2'},\ldots,\set{j_{\abs{e_2}-p-1}',j_{k_2}}},,\\
            \lab_2(e_1\cap e_2)&=\set{\set{j_{k_2},j_{(k_2 \mod m) + 1}},\ldots,\set{j_{(k_2+p-2\mod m) + 1},j_{(k_2+p-1\mod m) + 1}}}.
        \end{align*}
        Defining the bijection
        \begin{equation*}
            f\colon V_{G_{H,\lab_1}}\to V_{G_{H,\lab_2}}\colon v \mapsto\begin{cases}
                j_{(l+(k_2-k_1)-1 \mod m) +1},\quad &\text{if }\exists l\in\set{1,\ldots,m}\colon v=i_l,\\
                j_l',\quad &\text{if }\exists l\in\set{1,\ldots,\abs{e_2}-p-1}\colon v=i_l',
            \end{cases}
        \end{equation*}
        it can be verified that $f$ satisfies \eqref{eq:isomhypergraphs} and therefore, $G_{H,\lab_1}$ is isomorphic to $G_{H,\lab_2}$.
\end{proof}
\begin{proof}[Proof of Theorem \ref{thm:uniquenessdimtwo}]
    First, we note that there exists a fundamental cycle basis $B\in\BCS_{\LS_H}$ of $H$ (see \cite{GroYelAnd18}, for example). Hence, using that $\LS_H$ is connected, $\PG_B$ satisfies the assumption in Proposition \ref{prop:isomsimplegraphs12} and the remaining statement follows from Theorem \ref{cor:preimagefundbasis}.
\end{proof}
\subsection{Counterexamples for higher dimensions}
\label{subsec:counterexamples}
In this subsection, we will show some examples of non-isomorphic graphs whose cycle space is connect with dimension $n\geq 3$ and whose set of compiled hypergraphs are equal. This shows that the statement in Theorem \ref{thm:uniquenessdimtwo} is not true for $n\geq 3$. We start with the following example.
\begin{example}
    \label{examp:fourbodyplaqoneneigh}
    Let $H=(V,E)$ be the hypergraph with vertex set $V\coloneqq\set{1,2,\ldots,10}$ and edge set $E\coloneqq\set{\set{1,2,3,4},\set{2,5,6,7},\set{6,8,9,10}}$. Then, we choose the two loop labelings $\lab_1$ and $\lab_2$ according to Figure \ref{fig:example_looplab_1} and \ref{fig:example_looplab_2}. Since $\dim(\LS_{G_{H,\lab_1}})=\dim(\LS_{G_{H,\lab_2}})=3$ we deduce from Theorem \ref{thm:mainthm} that $[G_{H,\lab_1}],[G_{H,\lab_2}]\in\parmap\vert_{\DHGSE}^{-1}(\set{y})$ where $y\coloneqq\parmap([G_{H,\lab_1}])$. However, we observe that the degree of vertex $3$ in $H_1$ equals four whereas every vertex in $H_2$ has degree three at most. This implies that $G_{H_1,\lab_1}$ and $G_{H_2,\lab_2}$ are not isomorphic graphs. Figure \ref{fig:induced_simplegraph_1} and \ref{fig:induced_simplegraph_2} show the induced graphs of the loop labelings $\lab_1$ and $\lab_2$. We emphasize that the first loop labeling $\lab_1$ is constructed in such a way that vertex $3$ in the induced graph $G_{H,\lab_1}$ is contained in the labelings of the two vertices $2$ and $6$ in the four-body plaquette layout which connect all three plaquettes. Note that by iteratively appending new four-body plaquettes to the rightmost plaquette and extending the two labelings $\lab_1$ and $\lab_2$, counterexamples with an arbitrary number of plaquettes $n\geq 3$ can be constructed.
\end{example}
\begin{figure}[h!]
    \begin{minipage}{0.49\textwidth}
\begin{center}
	\includegraphics{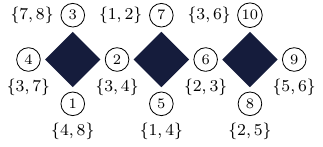}
    \captionof{figure}{Loop labeling $\lab_1$ in Example \ref{examp:fourbodyplaqoneneigh}.}
    \label{fig:example_looplab_1}
\end{center}
\end{minipage}
\hfill
    \begin{minipage}{0.49\textwidth}
\begin{center}
	\includegraphics{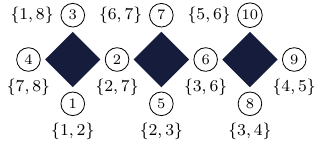}
    \captionof{figure}{Loop labeling $\lab_2$ in Example \ref{examp:fourbodyplaqoneneigh}.}
    \label{fig:example_looplab_2}
\end{center}
\end{minipage}\vspace{0.25cm}
    \begin{minipage}{0.49\textwidth}
\begin{center}
	\includegraphics{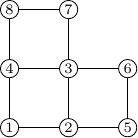}
    \captionof{figure}{Induced graph $G_{H,\lab_1}$ in Example \ref{examp:fourbodyplaqoneneigh}.}
    \label{fig:induced_simplegraph_1}
\end{center}
\end{minipage}
\hfill
    \begin{minipage}{0.49\textwidth}
\begin{center}
	\includegraphics{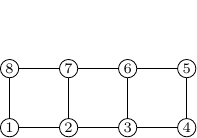}
    \captionof{figure}{Induced graph $G_{H,\lab_2}$ in Example \ref{examp:fourbodyplaqoneneigh}.}
    \label{fig:induced_simplegraph_2}
\end{center}
\end{minipage}
\end{figure}
\noindent As can be observed in previous example, each plaquette has at most only one common vertex with one other plaquette. This leads to a high degree of freedom in the choice of a loop labeling for such a plaquette layout from which we constructed two non-isomorphic graphs having the same compiled layout. In the next example, we present two non-isomorphic graphs which can be compiled to the same four-body plaquette layout where each plaquette can even have two common vertices with one other plaquette.
\begin{example}
    \label{examp:fourbodyplaqtwoneigh}
    We consider the hypergraph $H$ as illustrated in Figures \ref{fig:loop_labeling1_plaquette} and \ref{fig:loop_labeling2_plaquette} with seven edges and $16$ vertices and define the two loop labelings $\lab_1$ and $\lab_2$ as also presented in Figures \ref{fig:loop_labeling1_plaquette} and \ref{fig:loop_labeling2_plaquette}. From Theorem \ref{thm:dimformulals} we deduce that $\dim(\LS_{G_{H,\lab_1}})=\dim(\LS_{G_{H,\lab_2}})=7$ and therefore by Theorem \ref{thm:mainthm} we can conclude $[G_{H,\lab_1}],[G_{H,\lab_2}]\in\parmap\vert_{\DHGSE}^{-1}(\set{y})$ with $y\coloneqq\parmap([G_{H,\lab_1}])$. Similarly as in the previous example, we observe that the degree of vertex $3$ in $H_2$ equals six whereas every vertex in $H_1$ has degree five at most. This implies again that $G_{H_1,\lab_1}$ and $G_{H_2,\lab_2}$ are not isomorphic graphs. We note that in the edge $\set{5,8,9,10}$ we switched the labelings of the two vertices $9$ and $10$ in order to obtain a higher degree of the vertex $3$ in $G_{H,\lab_2}$.
\end{example}
\begin{figure}[h!]
    \begin{minipage}{0.49\textwidth}
\begin{center}
    \includegraphics{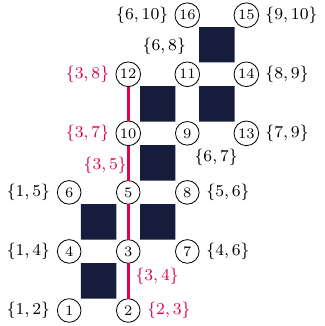}
    \captionof{figure}{Loop labeling $\lab_1$ in Example \ref{examp:fourbodyplaqtwoneigh}.}
    \label{fig:loop_labeling1_plaquette}
\end{center}
\end{minipage}
\hfill
    \begin{minipage}{0.49\textwidth}
\begin{center}
    \includegraphics{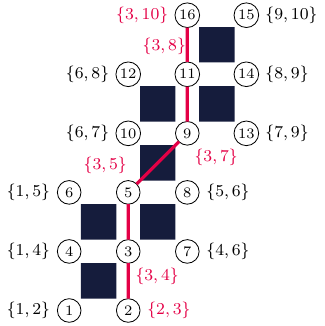}
    \captionof{figure}{Loop labeling $\lab_2$ in Example \ref{examp:fourbodyplaqtwoneigh}.}
    \label{fig:loop_labeling2_plaquette}
\end{center}
\end{minipage}
\end{figure}
\noindent As can be seen in Example \ref{examp:fourbodyplaqtwoneigh}, having a connected four-body plaquette layout where each plaquette has two common vertices with a subsequent plaquette is not sufficient to obtain only induced isomorphic graphs. However, there exists a class of compiled hypergraphs with four-body plaquettes for which a similar result as in Theorem \ref{thm:uniquenessdimtwo} holds. This class will be discussed in the last section of this article.
\subsection{Rectangular plaquette layouts}
\label{subsec:charrectplaqlayouts} 
Lastly, we focus on compiled hypergraphs which are so-called \emph{rectangular plaquette layouts} and defined as follows.
\begin{definition}
    Let $H=(V,E)$ be a hypergraph. We say that \emph{$H$ is a rectangular plaquette layout} if and only if there exist natural numbers $m,n\geq 2$ and an enumeration of the vertex set \[V=\set{v_{i,j}\mid 1\leq i\leq m,1\leq j\leq n}\] as well as an enumeration of the edge set \[E=\set{e_{i,j}\mid 1\leq i\leq m-1,1\leq j\leq n-1}\] such that
    \begin{equation*}
        \forall 1\leq i\leq m-1\colon\forall 1\leq j\leq n-1\colon e_{i,j}=\set{v_{i,j},v_{i,j+1},v_{i+1,j+1},v_{i+1,j}}
    \end{equation*}
    We denote $m$ as the number of vertical vertices and $n$ as the number of horizontal vertices of $H$. Moreover, we define $\PL_{\rectangle}^{(4)}\coloneqq \set{H\in\HGS\mid H\text{ is a rectangular plaquette layout}}$ as the set of all rectangular plaquette layouts.
\end{definition}
\noindent Figure \ref{fig:rectangularplaquette} shows an illustration of an example of a rectangular plaquette layout with five horizontal and four vertical nodes.
\begin{figure}[h!]
    \begin{minipage}{0.49\textwidth}
\begin{center}
	\includegraphics{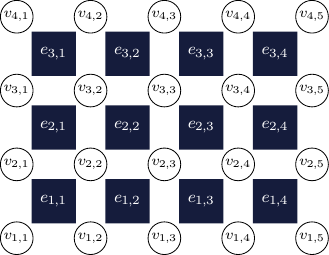}
    \captionof{figure}{A rectangular plaquette layout.}
    \label{fig:rectangularplaquette}
\end{center}
\end{minipage}
\hfill
    \begin{minipage}{0.49\textwidth}
     \vspace{0.4cm}
\begin{center}
	\includegraphics{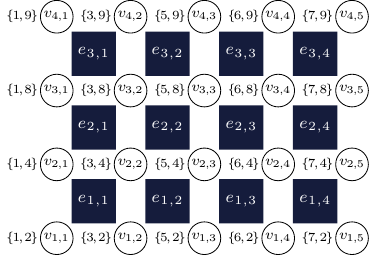}
    \captionof{figure}{Loop labeling of a rectangular plaquette layout.}
    \label{fig:looplbalrectplaq}
\end{center}
\end{minipage}
\end{figure}
\subsubsection{Characterization of preimages of rectangular plaquette layouts}
\noindent The last theorem of this section presents an equivalent condition for a graph $H$ that can be compiled on a rectangular plaquette layout.
\begin{theorem}
    \label{thm:preimagerectangularlayouts}
    Let $m,n\geq 2$ and $H=(V,E)$ be a graph. Then, there exists $B\in\BCS_{\LS_{H}}
    $ such that the compiled hypergraph $\PG_B\in\PL_{\rectangle}^{(4)}$ is a rectangular plaquette layout with $m$ vertical and $n$ horizontal vertices if and only if $H\vert_{E_{\LS}}$ is a complete bipartite graph with two partitions of size $m$ and $n$.
\end{theorem}
\noindent Figure \ref{fig:compiledbibartitegraph} shows a complete bibartite graph with two partitions of size five and four.\\
\begin{figure}
    \centering
    \includegraphics{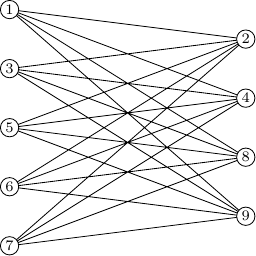}
    \caption{A complete bipartite graph which can be compiled to the layout in Figure \ref{fig:rectangularplaquette} by choosing a basis similarly as in the proof of Theorem \ref{thm:preimagerectangularlayouts}.}
    \label{fig:compiledbibartitegraph}
\end{figure}

\noindent For better readability, we again divide the proof into several parts. Similarly as in Theorem \ref{thm:uniquenessdimtwo}, we first prove that any two loop labelings of rectangular plaquette layouts induce graphs which are isomorphic.
\begin{prop}
    \label{prop:isominducsimplgraphsrpl}
    Let $H=(V,E)\in\PL_{\rectangle}^{(4)}$ be a be a rectangle plaquette layout. Moreover, let $\lab_1,\lab_2$ be two loop labelings of $H$ with $\dim(G_{H,\lab_1})=\dim(G_{H,\lab_2})=\abs{E}$. Then $G_{H,\lab_1}$ is isomorphic to $G_{H,\lab_2}$.
\end{prop}
\noindent For the proof of Proposition \ref{prop:isominducsimplgraphsrpl}, we need the following lemma, which characterizes loop labelings of rectangular plaquette layouts. Figure \ref{fig:looplbalrectplaq} shows how the plaquette layout has to be labeled according to conditions \ref{item:prop1lopplabrpl}--\ref{item:prop4lopplabrpl} in the subsequent lemma.
\begin{lemma}
    \label{lem:condlooplabrpl}
    Let $H=(V,E)\in\PL_{\rectangle}^{(4)}$ be a rectangular plaquette layout with $m\geq 3$ vertical and $n\geq 3$ horizontal nodes. Moreover, let $\lab$ be a loop labeling of $H$ with $\dim(G_{H,\lab})=\abs{E}$. Then, there exist pairwise distinct numbers $i_1,\ldots,i_{m+n}\in\N$ such that
    \begin{enumerate}
        \item \label{item:prop1lopplabrpl} $\lab(v_{1,1})=\set{i_1,i_2}$, $\lab(v_{1,2})=\set{i_2,i_3}$, $\lab(v_{2,2})=\set{i_3,i_4}$ and $\lab(v_{2,1})=\set{i_1,i_4}$.
        \item \label{item:prop2lopplabrpl} $\forall 2\leq j\leq n-1\colon \lab(v_{1,j+1})=\set{i_2,i_{j+3}}\text{ and }\lab(v_{2,j+1})=\set{i_{j+3},i_4}$.
        \item \label{item:prop3lopplabrpl} $\forall 2\leq i\leq m-1\colon \lab(v_{i+1,1})=\set{i_1,i_{n+i+1}}\text{ and }\lab(v_{i+1,2})=\set{i_3,i_{n+i+1}}$
        \item \label{item:prop4lopplabrpl} $\forall 2\leq i\leq m-1\colon\forall 2\leq j\leq n-1\colon  \lab(v_{i+1,j+1})=\set{i_{n+i+1},i_{j+3}}$
    \end{enumerate}
\end{lemma}
\begin{proof}
    First, we note that $\lab(e)\in \LSP_{G_{H,\lab}}$ for all $e\in E$ since the number of vertices of each edge is equal to four. In the following, we choose two enumerations of the edge set $E$
    \begin{equation*}
        e_k^{(1)}\coloneqq e_{a_{k,1},b_{k,1}},\quad a_{k,1}\coloneqq (k-1)//(n-1)+1,\ b_{k,1}\coloneqq (k-1)\mod (n-1)+1,
    \end{equation*}
    the ``horizontal'' enumeration, and
    \begin{equation*}
        e_k^{(2)}\coloneqq e_{a_{k,2},b_{k,2}},\quad a_{k,2}\coloneqq (k-1)\mod (m-1)+1,\ b_{k,2}\coloneqq (k-1)//(m-1)+1,\,
    \end{equation*}
    the ``vertical'' enumeration of $E$, where $k=1,\ldots,(m-1)(n-1)$. Here, $//$ denotes the floor division of integers. Then, we observe for $i\in\set{1,2}$
    \begin{equation*}
        \forall 2\leq k\leq (m-1)(n-1)\colon e_k^{(i)}\setminus\bigcup_{l=1}^{k-1} e_l^{(i)}\neq \emptyset\land e_k^{(i)}\cap\bigcup_{l=1}^{k-1} e_l^{(i)}\neq\emptyset.
    \end{equation*}
    Thus, Lemma \ref{lem:restrlooplab} \ref{item:restrlooplab} shows that for all $1\leq k\leq (m-1)(n-1)$ and $i\in\set{1,2}$
    \begin{equation}
        \label{eq:dimpropertylooplab}
        \dim(\LS_{G_{H_{k,i},\lab_{k,i}}})=k,\quad H_{k,i}\coloneqq (V_{k,i},{E_{k,i}})\quad\text{and}\quad\lab_{k,i}\coloneqq \lab\vert_{V_{k,i}},
    \end{equation}
    where $E_{k,i}\coloneqq\bigcup_{l=1}^{k}\set{e_l^{(i)}}$ and $V_{k,i}\coloneqq\set{v\in V\mid \exists 1\leq l\leq k\colon v\in e_l^{(i)}}$.
    \begin{enumerate}[wide=\parindent,leftmargin=0pt,align=left]
        \item \label{item:proofprop1lopplabrpl} Next, applying Lemma \ref{lem:outerloops} \ref{item:innerpath} with $H_{2,1}$ and $E_{2,1}=\set{e_1^{(1)},e_2^{(1)}}=\set{e_{1,1},e_{1,2}}$ in \eqref{eq:dimpropertylooplab}, we can imply that there exist pairwise distinct numbers $i_2,i_3,i_4\in\N$ such that $\lab(v_{1,2})=\set{i_2,i_3}$, $\lab(v_{2,2})=\set{i_3,i_4}$. Furthermore, using Lemma \ref{lem:outerloops} \ref{item:innerpath} again with $H_{2,2}$ and $E_{2,2}=\set{e_1^{(2)},e_{2}^{(2)}}=\set{e_{1,1},e_{2,1}}$ in \eqref{eq:dimpropertylooplab} as well as $\lab(e_1)\in\LSP_{G_{H,\lab}}$, there exists $i_1\in\N\setminus\set{i_2,i_3,i_4}$ with $\lab(v_{2,1})=\set{i_1,i_4}$ and $\lab(v_{1,1})=\set{i_1,i_2}$.
        \item \label{item:proofprop2lopplabrpl} Now, we fix the elements $i_1,\ldots,i_4\in\N$ in \ref{item:proofprop1lopplabrpl} and show by induction that there exist for all $2\leq k\leq n-1$ pairwise distinct numbers $i_5,\ldots,i_{k+3}\in\N\setminus\set{i_1,i_2,i_3,i_4}$ with
        \begin{equation}
            \label{eq:proofprop2lopplabrpl} 
            \forall 2\leq j\leq k\colon \lab(v_{1,j+1})=\set{i_2,i_{j+3}}\text{ and }\lab(v_{2,j+1})=\set{i_4,i_{j+3}}.
        \end{equation}
        For $k=2$, applying Lemma \ref{lem:outerloops} \ref{item:outerpath} with $H_{2,1}$ shows that there exists $i_5\in\N\setminus\set{i_1,i_2,i_3,i_4}$ such that $\lab(\set{v_{1,3},v_{2,3}})=\set{\set{i_2,i_5},\set{i_5,i_4}}$. Since $e_m^{(2)}\cap e_{m+1}^{(2)}=e_{1,2}\cap e_{2,2}=\set{v_{2,2},v_{2,3}}$, we obtain by \ref{item:proofprop1lopplabrpl} as well as applying Lemma \ref{lem:outerloops} \ref{item:innerpath} on $H_{m+1,2}$ the relations $\lab(v_{1,3})=\set{i_2,i_5}$ and $\lab(v_{2,3})=\set{i_4,i_5}$. The remaining statement follows by induction and the same arguments.
        \item In the last step of the proof, we apply \ref{item:proofprop1lopplabrpl} and \ref{item:proofprop2lopplabrpl} and choose pairwise distinct $i_1,\ldots,i_{n+2}\in\N$ satisfying the first two conditions in Lemma \ref{lem:condlooplabrpl}. We finally show by induction that for all $2\leq k\leq m-1$ there exist pairwise distinct $i_{n+3},\ldots,i_{n+k+1}\in\N\setminus\set{i_1,\ldots,i_{n+2}}$ satisfying
        \begin{align}
            \forall 2\leq i\leq k\colon &\lab(v_{i+1,1})=\set{i_1,i_{n+i+1}}\text{ and }\lab(v_{i+1,2})=\set{i_3,i_{n+i+1}}\label{eq:proofprop3alopplabrpl}\\
            \forall 2\leq i\leq k\colon\forall 2\leq j\leq n-1\colon  &\lab(v_{i+1,j+1})=\set{i_{n+i+1},i_{j+3}}.\label{eq:proofprop3blopplabrpl}
        \end{align}
        Suppose that $k=2$. We use the same arguments as before and apply Lemma \ref{lem:outerloops} \ref{item:outerpath} on $H_{2,2}$ to find $i_{n+3}\in\N\setminus\set{i_1,\ldots,i_{n+2}}$ with $\lab(\set{v_{3,1},v_{3,2}})=\set{\set{i_1,i_{n+3}},\set{i_3,i_{n+3}}}$. Moreover, applying Lemma \ref{lem:outerloops} \ref{item:innerpath} on $H_{n+1,1}$ and using the first two conditions in Lemma \ref{lem:condlooplabrpl}, we deduce that $\lab(v_{3,2})=\set{i_3,i_{n+3}}$ and therefore $\lab(v_{3,1})=\set{i_1,i_{n+3}}$, which shows that \eqref{eq:proofprop3alopplabrpl} holds for $k=2$. Since $e_{2,2}=\set{v_{2,2},v_{2,3},v_{3,3},v_{3,2}}$ we deduce from $\lab(e_{2,2})\in \LSP_{G_{H,\lab}}$ and the relations $\lab(v_{2,2})=\set{i_3,i_4}$, $\lab(v_{3,2})=\set{i_3,i_{n+3}}$, $\lab(v_{2,3})=\set{i_4,i_5}$ that $\lab(v_{3,3})=\set{i_5,i_{n+3}}$. This shows that \eqref{eq:proofprop3blopplabrpl} holds for $j=2$. Now, suppose $n\geq 4$ and \eqref{eq:proofprop3blopplabrpl} holds for some particular $2\leq j<n-1$ and $i=2$. By using the induction hypothesis, we observe that $\lab(v_{3,j+1})=\set{i_{n+3},i_{j+3}}$. Since $\lab(v_{2,j+1})=\set{i_4,i_{j+3}}$ and $\lab(v_{2,j+2})=\set{i_4,i_{j+4}}$ by \eqref{eq:proofprop2lopplabrpl}, and $\lab(e_{2,j+1})\in \LSP_{G_{H,\lab}}$ we therefore obtain $\lab(v_{3,j+2})=\set{i_{n+3},i_{j+4}}$.\\
        Assume now that $m\geq 4$ and the above statement with conditions \eqref{eq:proofprop3alopplabrpl} and \eqref{eq:proofprop3blopplabrpl} holds for a particular $2\leq k<m-1$, and let $i_{n+3},\ldots,i_{n+k+1}\in\N\setminus\set{i_1,\ldots,i_{n+2}}$ be some integers satisfying these two conditions. In order to show the statement for $k+1$, we are only left to show that there exists $i_{n+k+2}\in\N\setminus\set{i_1,\ldots,i_{n+k+1}}$ which fulfills conditions \eqref{eq:proofprop3alopplabrpl} and \eqref{eq:proofprop3blopplabrpl} for $i=k+1$. This again follows by induction and the same arguments.\qedhere
    \end{enumerate}
\end{proof}
\begin{proof}[Proof of Proposition \ref{prop:isominducsimplgraphsrpl}]
    \begin{enumerate}[wide=\parindent,leftmargin=0pt,align=left]
        \item We start by showing the statement for $m=2$ and $n>2$. The case $m\geq 2$ and $n=2$ can be proven analogously.\\
        First, we show that for any loop labeling $\lab$ of $H$ with $\dim(\LS_{G_{H,\lab}})=\abs{E}$ there exists for all $3\leq k\leq n$ pairwise $i_1,\ldots,i_{k+2}\in\N$ such that
        \begin{enumerate}[label=(\alph*)]
            \item \label{item:looplablinecond1} $\lab(\set{v_{1,1},v_{2,1}})=\set{\set{i_1,i_2},\set{i_1,i_4}}$,
            \item \label{item:looplablinecond2} $\lab(\set{v_{1,2},v_{2,2}})=\set{\set{i_2,i_3},\set{i_3,i_4}}$ and
            \item \label{item:looplablinecond3} $\forall 3\leq l\leq k\colon \lab(\set{v_{1,l},v_{2,l}})=\set{\set{i_2,i_{l+2}},\set{i_{l+2},i_4}}$.
        \end{enumerate}
        For $k=3$, we observe from the application of Lemma \ref{lem:outerloops} \ref{item:innerpath} on $H_2=(V_2,E_2)$ with $E_2=\set{e_{1,1},e_{1,2}}$ and $V_2=\set{v\in V\mid v\in E_2}$ that there exist pairwise distinct $i_2,i_3,i_4\in\N$ such that $\lab(\set{v_{1,2},v_{2,2}})=\set{\set{i_2,i_3},\set{i_3,i_4}}$. Moreover, since $\lab(e_{1,1}),\lab(e_{1,2})\in \LSP_{G_{H,\lab}}$ we find by Lemma \ref{lem:restrlooplab} \ref{item:restrlooplab} distinct numbers $i_1,i_5\in\N\setminus\set{i_2,i_3,i_4}$ satisfying $\lab(\set{v_{1,1},v_{2,1}})=\set{\set{i_1,i_2},\set{i_1,i_4}}$ and $\lab(\set{v_{1,3},v_{2,3}})=\set{\set{i_2,i_5},\set{i_4,i_5}}$. Therefore, conditions \ref{item:looplablinecond1}, \ref{item:looplablinecond2} and \ref{item:looplablinecond3} hold for $k=3$. Now, suppose $n\geq 4$ and the statement holds for some particular $3\leq k<n$. Since $\lab(\set{v_{1,k},v_{2,k}})=\set{\set{i_2,i_{k+2}},\set{i_{k+2},i_4}}$ we can find by Lemma \ref{lem:outerloops} \ref{item:outerpath} a positive integer $i_{k+3}\in\N\setminus\set{i_1,\ldots,i_{k+2}}$ satisfying the relation ${\lab(\set{v_{1,{k+1}},v_{2,{k+1}}})=\set{\set{i_2,i_{k+3}},\set{i_{k+3},i_4}}}$. Thus, by our induction hypothesis \ref{item:looplablinecond1}, \ref{item:looplablinecond2} and \ref{item:looplablinecond3} hold for $k+1$, which shows the desired statement. As a consequence, we can find for the loop labelings $\lab_1$ and $\lab_2$ pairwise distinct numbers $i_1,\ldots,i_{n+2}\in\N$ and $j_1,\ldots,j_{n+2}\in\N$, respectively, satisfying conditions \ref{item:looplablinecond1}--\ref{item:looplablinecond3} for $k=n$. Thus, we see that $V_{G_{H,\lab_1}}=\set{i_1,\ldots,i_{n+2}}$ as well as $V_{G_{H,\lab_2}}=\set{j_1,\ldots,j_{n+2}}$. Moreover, from the conditions \ref{item:looplablinecond1}--\ref{item:looplablinecond3} we immediately obtain that the function
        $f\colon V_{G_{H,\lab_1}}\to V_{G_{H,\lab_2}}$ defined by $f(i_k)\coloneqq j_k$ for $k=1,\ldots,n+2$ fulfills property \eqref{eq:isomhypergraphs}, which shows that the induced graphs $G_{H,\lab_1}$ and $G_{H,\lab_2}$ are isomorphic.
        \item Now, let $H\in\PL_{\rectangle}^{(4)}$ be a rectangular plaquette layout with $m\geq 3$ vertical and $n\geq 3$ horizontal nodes. By Lemma \ref{lem:condlooplabrpl}, we can find pairwise distinct numbers $i_1,\ldots,i_{m+n}\in\N$ and $j_1,\ldots,j_{m+n}\in\N$ such that conditions \ref{item:prop1lopplabrpl}--\ref{item:prop4lopplabrpl} in Lemma \ref{lem:condlooplabrpl} for $\lab_1$ and $\lab_2$ are satisfied. Finally, we observe that $V_{G_{H,\lab_1}}=\set{i_1,\ldots,i_{m+n}}$ as well as $V_{G_{H,\lab_2}}=\set{j_1,\ldots,j_{m+n}}$ and deduce $f(\lab_1(v_{i_j}))=\lab_2(v_{i_j})$ for $i=1,\ldots,m$ and $j=1,\ldots,n$, where we used again the bijection function $f\colon V_{G_{H,\lab_1}}\to V_{G_{H,\lab_2}}$ with $f(i_k)=j_k$ for $k=1,\ldots,m+n$. This proves the remaining statement for the case $m,n\geq 3$.\qedhere
    \end{enumerate}
\end{proof}
\begin{proof}[Proof of Theorem \ref{thm:preimagerectangularlayouts}]
    \begin{enumerate}[wide=\parindent,leftmargin=0pt,align=left]
        \item \label{item:existrectlay} First, define $H\coloneqq(V,E)$ by $V\coloneqq V_1\cup V_2$ and\\
    $E\coloneqq\set{\set{v_1,v_2}\mid v_1\in V_1,v_2\in V_2}$, where $V_1\coloneqq\set{i_1,\ldots,i_m}$, $V_2\coloneqq\set{j_1,\ldots,j_n}$ and
    \begin{equation*}
       i_k\coloneqq\begin{cases}
           1,\quad &\text{if $k=1$},\\
           3,\quad &\text{if $k=2$},\\
           k+2,\quad &\text{else},
        \end{cases}\qquad\qquad
        j_l\coloneqq\begin{cases}
           2,\quad &\text{if $l=1$},\\
           4,\quad &\text{if $l=2$},\\
           l+m,\quad &\text{else},
        \end{cases}
    \end{equation*}
    for $k=1,\ldots,m$ and $l=1,\ldots,n$. By definition, we obtain that $H$ is a complete bipartite graph with the two partitions $V_1$ and $V_2$ satisfying $\abs{V_1}=m$ and $\abs{V_2}=n$ and $V_1\cap V_2=\emptyset$. Next, we define
    \begin{equation*}
        a_k\coloneqq (k-1)\mod (n-1)+1\quad
        b_k\coloneqq (k-1)//(n-1)+1,\quad \text{for $k=1,\ldots,(m-1)(n-1)$}
    \end{equation*}
    and \[C_k\coloneqq\set{\set{i_{a_k},j_{b_k}},\set{i_{a_k},j_{b_{k}+1}},\set{i_{a_{k}+1},j_{b_{k}+1}},\set{i_{a_{k}+1},j_{b_k}}}\in \LSP_{H},\quad k=1,\ldots,(m-1)(n-1),\] and observe that \eqref{eq:weaklyfundbasis} holds for all $k=2,\ldots,(m-1)(n-1)$ and thus, $B\coloneqq\set{C_1,\ldots,C_{(m-1)(n-1)}}$ is linearly independent in $\LS_{H}$. Moreover, by Theorem \ref{thm:dimformulals} we obtain $\dim(\LS_{H})=\abs{E}-\abs{V}+1=mn-(m+n)+1=(m-1)(n-1)$ and therefore, $B\in\BCS_{\LS_{H}}$. By replacing $m$ with $m+1$ and $n$ with $n+1$ in the definition $a_k$ and $b_k$, we deduce that $h\colon\set{1,\ldots,mn}\to E_B\colon k\mapsto\set{i_{a_k},j_{b_k}}$ is an enumeration of $E_{B}$. By using $a_k$ and $b_k$ again to re-enumerate the vertices $V_{\PG_{B}}$ as well as $E_{\PG_{B}}$, we see that $\PG_{B}$ is a rectangular plaquette layout with $m$ vertical and $n$ horizontal nodes. This observation and Lemma \ref{lem:welldefinedparmap} shows that for any hypergraph $H'$ for which $H'\vert_{E_{\LS_H}}$ is a complete bipartite graph that there exists $B'\in\BCS_{\LS_{H'}}$ such that $\PG_{B'}\in\PL_{\rectangle}^{(4)}$ is a rectangular plaquette layout with $m$ vertical and $n$ horizontal nodes.
    \item Now suppose that $H$ is a complete bipartite graph with two partitions of size $m$ and $n$. By \ref{item:existrectlay} there exists $B\in\BCS_{\LS_{H}}$ with $B\subset \LSP_H$ and property \eqref{eq:weaklyfundbasis} such that $\PG_B\in\PL_{\rectangle}^{(4)}$ is a rectangular plaquette layout with $m$ vertical and $n$ horizontal nodes. Define $y\coloneqq\parmap([H])$ and let $H'=(V',E')$ be another hypergraph for which there exists $B'\in\BCS_{\LS_{H}}$ such that the compiled hypergraph $\PG_{B'}\in\PL_{\rectangle}^{(4)}$ is also a rectangular plaquette layout with $m$ horizontal and $n$ vertical nodes. This implies that $\PG_{B'}\cong\PG_{B}$ and thus by Lemma \ref{lem:condeqcomphyp} we obtain $\parmap([H])=\parmap([H'])$, or equivalently, $[H']\in\parmap\vert_{\DHGSE}^{-1}([H])$. Applying Corollary \ref{cor:preimagefundbasis} on $H$ and using Proposition \ref{prop:isominducsimplgraphsrpl} yields $H\cong H'\vert_{E'_{\LS}}$ and therefore, $H'\vert_{E'_{\LS}}$ is a bipartite graph with two partitions of size $m$ and $n$, which shows the claimed statement.\qedhere
    \end{enumerate}
\end{proof}
\subsubsection{Polynomial-time compilation on rectangular plaquette layouts}
In the last section of this article, we apply our results from the previous section to derive a polynomial-time algorithm that efficiently determines whether an optimization problem can be mapped onto a rectangular plaquette layout under the parity transformation while also constructing the corresponding physical layout. We emphasize again that this is highly relevant for practical applications. On the one hand, local qubit connectivity is a common feature of quantum hardware, and on the other hand, square-grid-like connectivity graphs are widely present across many hardware platforms.

\noindent As we have already seen in the previous section, we verified that
\begin{multline*}
    \set{H\in\HGS\,\left\vert\, \text{$H$ is a graph and there exists $B\in\BCS_{\LS_{H}}$ such that $\PG_B\in\PL_{\rectangle}^{(4)}$}\right.}\\
    =\set{H=(V,E)\in\HGS\,\left\vert\, \text{$H\vert_{E_{\LS}}$ is a complete bipartite graph}\right.},
\end{multline*}
or in other words, a given optimization problem that is represented as a graph can be compiled on a rectangular plaquette layout if and only if the graph restricted to the set of edges which are contained in an element of the cycle space is a complete bipartite graph. Hence, this necessary and sufficient condition immediately allows for the application of well-established algorithms for solving this computational problem in polynomial time. For example, in order to compute the restriction of a given graph $H=(V,E)\in\HGS$ onto ${E_{\LS}}$ one could first apply forest growing algorithms \cite{GroYelAnd18}, which usually have a time complexity of $\bigo(\abs{V}+\abs{E})$ and compute a spanning forest of $H$. If the number of components of the graph $H$ is one, then iteratively collecting the unique fundamental cycles generates the restricted graph $H\vert_{E_{\LS}}$. Note that each fundamental cycle can be determined by adding a non tree-edge to the spanning tree and calculating the two unique paths which start from the endpoints of the non tree-edge, traversing through their parents and end in the first common node. The running time of this part of the algorithm is proportional to $\abs{V}\abs{E}-\abs{V}^2+\abs{V}$, since the number of non tree-edges equals the dimension of the cycle space which is $\abs{E}-\abs{V}+1$. As a second step, by using classical search algorithms such as breadth-first or depth-first \cite{GroYelAnd18}, which are also being used in forest growing algorithms, two independent partitions (in case the previously calculated restricted graph $H\vert_{E_{\LS}}$ is bipartite) in $\bigo(\abs{V}+\abs{E})$ running time can be computed. Finally, checking whether the number of neighbors of each vertex in the first partition equals the number of elements of the second partition, completes the whole algorithm and yields a total time complexity proportional to $\abs{V}\abs{E}-\abs{V}^2+\abs{V}$. The corresponding cycle basis of the complete bipartite graph which compiles to a rectangular plaquette layout is given at the beginning of the proof of Theorem \ref{thm:preimagerectangularlayouts}. Note that this algorithm can easily be extended to non-connected graphs.
\section*{Conclusion and future research}
In this article, we established the parity transformation as a mapping between equivalence classes of hypergraphs and sets of equivalence classes of hypergraphs which can be employed to create physical layouts of a quantum device. By introducing so-called loop labelings, we derived theoretical results for the description of the preimage of a given set of equivalence classes of compiled hypergraphs when equivalence classes of graphs are being considered. Using these results, we then could derive equivalent conditions for the parity mapping being injective on any subdomain of this set of equivalences classes. By presenting various examples of non-isomorphic graphs which do have the same set of physical layouts, we showed that the parity transformation is in general not a unique mapping. This demonstrates that the same physical quantum device can be used to solve different optimization problems, thereby highlighting an additional benefit of this encoding strategy.
As a further application of our theoretical derivations, we introduced an algorithm that efficiently compiles optimization problems onto a rectangular plaquette layout, a common configuration on quantum hardware platforms, under the parity transformation. It is important to note that, while determining all different bases and checking if they correspond to rectangular plaquette layouts is an exponential task for most graphs, our approach reduces this process to a polynomial problem. This reduction makes the compilation process feasible, even for optimization problems involving a large number of variables and interactions.

In the future, we want to derive further uniqueness results for plaquette layouts which are proper subhypergraphs of rectangular plaquette layouts.
As we already have seen in Subsection \ref{subsec:counterexamples}, there exist such subhypergraphs whose corresponding graphs in the domain of the  parity transformation are not uniquely determined. We therefore seek for specific restrictions for plaquette layouts with local four-body interactions and verify that the preimage of the corresponding set of physical layouts is uniquely determined by
similarly examining induced graphs of loop labelings. As a further step, we aim for characterizing the preimage of sets of physical layouts which contain any subhypergraph of a rectangular plaquette layout. In addition, we also intend to provide similar algorithms as for rectangular plaquette layouts which solve the analogous computational problem for this broader class of physical layouts. On the one hand, thoroughly characterizing these preimages not only enhances deeper insights into the complexities of the parity transformation but also establishes a direct connection between graphs and compiled physical layouts with local four-body constraints. On the other hand, it allows the opportunity to investigate further compilation strategies and algorithms that may lead to a reduction in running time for parity compilation, where the challenge is to transform any optimization problem to a physical layout with local three- and four-body plaquettes with ancilla qubits. Furthermore, extending our results in Section \ref{sec:uniquenesssimplegraph} from graphs to hypergraphs is also of great interest,
paving the way to similarly characterize the preimage of single sets of physical layouts that include arbitrary multi-body optimization problems.

\section*{Acknowledgements}
The authors would like to thank Anette Messinger and Christoph Fleckenstein for the detailed review of the manuscript as well as  their crucial feedback. This work was supported by the Austrian Research Promotion Agency (FFG Project No. FO999896208). This research was funded in whole, or in part, by the Austrian Science Fund (FWF) SFB BeyondC Project No. F7108-N38 (DOI: 10.55776/F71), through a START grant under Project No. Y1067-N27 (DOI: 10.55776/Y1067) and I 6011 (DOI: 10.55776/I6011). For the purpose of open access, the author has applied a CC BY public copyright licence to any Author Accepted Manuscript version arising from this submission. This project was funded within the QuantERA II Programme that has received funding from the European Union's Horizon 2020 research and innovation programme under Grant Agreement No. 101017733 (DOI: 10.3030/101017733).

\end{document}